\documentclass[12pt]{article}

\usepackage{graphicx}
\usepackage{mathtools}
\usepackage{bm}
\usepackage{amssymb}
\usepackage{amsmath}
\usepackage{fullpage}
\usepackage{tikz}
\usepackage{subcaption}
\usepackage[colorlinks=true,linkcolor=blue,citecolor=red,plainpages=false,pdfpagelabels]{hyperref}

\newcommand{\bra}[1]{\langle {#1} |}
\newcommand{\ket}[1]{| {#1} \rangle}
\newcommand{\Tr}{\operatorname{Tr}}
\newtheorem{theorem}{Theorem}

\newtheorem{lemma}[theorem]{Lemma}
\newtheorem{proposition}[theorem]{Proposition}
\newtheorem{remark}{Remark}
\newenvironment{proof}[1][Proof]{\noindent\textbf{#1.} }{\ \rule{0.5em}{0.5em}}

\allowdisplaybreaks

\begin{document}

\title{RLD Fisher Information Bound for Multiparameter Estimation of Quantum Channels}
\author{Vishal Katariya\thanks{Hearne Institute for Theoretical Physics, Department of Physics and Astronomy, and Center for Computation and Technology, Louisiana State University, Baton Rouge, Louisiana 70803, USA} \and Mark M.~Wilde\footnotemark[1] \thanks{Stanford Institute for Theoretical Physics, Stanford University, Stanford, California 94305, USA}}
\date{\today}

\maketitle

\begin{abstract}
	One of the fundamental tasks in quantum metrology is to estimate multiple parameters embedded in a noisy process, i.e., a quantum channel. In this paper, we study fundamental limits to quantum channel estimation via  the concept of amortization and the right logarithmic derivative (RLD) Fisher information value.
	Our key technical result is the proof of a chain-rule inequality for the RLD Fisher information value, which implies that amortization, i.e., access to a catalyst state family, does not increase the RLD Fisher information value of quantum channels. This technical result leads to a  fundamental and efficiently computable limitation for multiparameter channel estimation in the sequential setting, in terms of the RLD Fisher information value.
	As a consequence, we conclude that if the RLD Fisher information value is finite, then Heisenberg scaling is unattainable in the multiparameter setting.
\end{abstract}

\tableofcontents

\section{Introduction}

Parameter estimation is an information-processing task in which quantum technologies can provide an improvement in performance over the best known classical technologies. On one hand, classical parameter estimation is fundamentally limited by the shot-noise limit, which means that the error scales no better than $1/n$. Here, $n$ refers either to the number of channel uses or to the total probing time allowed in the estimation task. On the other hand, quantum parameter estimation can offer a superclassical ``Heisenberg'' scaling in estimation error, in principle~\cite{Giovannetti2006}. Heisenberg scaling refers to the mean-squared error of an estimator scaling as $1/n^2$. One of the fundamental questions in quantum estimation theory is to identify estimation tasks for which Heisenberg scaling is possible, and attainable.

One of the most important mathematical tools in estimation theory, both classical and quantum, is the Cramer--Rao bound (CRB) \cite{Cram46, Rao45, Kay93}. It is used to place lower bounds on the mean squared error of estimators and involves an information  quantity known as the Fisher information. The latter captures the amount of information carried by a distribution, quantum state, or quantum channel regarding the unknown parameter(s). When using quantum resources, there is an infinite number of noncommutative generalizations of the Fisher information.

In this paper, we focus on one particular quantum generalization of the Fisher information, known as the right logarithmic derivative (RLD) Fisher information \cite{YL73}, and we study its properties and application in multiparameter quantum channel estimation. In particular, we prove that the RLD Fisher information of quantum channels is a single-letter and efficiently computable multiparameter Cramer--Rao bound for all quantum channels. 
Our bound applies to all quantum channels in the most general setting of sequential channel estimation, depicted in Figure~\ref{fig:sequential strategies}, which encompasses all possible noisy dynamics. This means that our bound is universally applicable. Moreover, it is ``single-letter'', a term from information theory, which means that the RLD Fisher information is evaluated with respect to a single copy of the channel only (see Theorem \ref{thm:single-letter-multi-param-crb} for a precise statement of our result). Our bound thus has a number of desirable properties as well as numerical amenability, while being applicable to sequential, $n$-round estimation strategies.

We approach this problem by introducing the amortized RLD Fisher information value, which quantifies the net increase in RLD Fisher information that one obtains by sampling from a quantum channel. We prove a chain-rule property for the RLD Fisher information value of a quantum channel, which we use to conclude our single-letter multiparameter Cramer--Rao bound.
This bound has the implication that Heisenberg scaling is unattainable for channel estimation in the multiparameter setting, whenever the RLD Fisher information value is finite. 
Finally, we apply our bound to a concrete example of physical interest: estimating the parameters of a generalized amplitude damping channel. 

\begin{figure*}
\centering
\begin{tikzpicture}[scale=1.5]
	\draw (0.375, 0.75-0.375) -- (0.75,0);
	\draw (0.375, 0.75-0.375) -- (0.75, 1.5-0.75);
	\node at (0.15 - 0.15, 0.75-0.375) {$\rho_{R A_1}$};
	\draw (0.75,0) -- (1,0);
	\draw (1,-0.25) rectangle (1.5,0.25);
	\node at (1.25,0) {$\mathcal{N}^{\bm{\theta}}$};
	\node at (0.5,0) {\small $A_1$};
	\node at (1.8,0.14 + 0.04) {\small$B_1$};
	\draw (1.5,0) -- (2,0);
	\draw (0.75,1.5-0.75) -- (2,1.5-0.75);
	\draw (2, 0.75+0.25) rectangle (2.5, 0.75-1); 
	\draw (2.5, 1.5-0.75) -- (4, 1.5-.75);
	\node at (2.25, 0.75-0.375) {$\mathcal{S}^1$};
	\draw (3,-0.25) rectangle (3.5,0.25);
	\node at (3.25,0) {$\mathcal{N}^{\bm{\theta}}$};
	\node at (2.75,0.14 + 0.04) {\small$A_2$};
	\node at (3.75,0.14 + 0.04) {\small$B_2$};
	\draw (3.5,0) -- (4,0); 
	\draw (2.5,0) -- (3,0); 
	\draw (4.5, 0) -- (5, 0); 
	\node at (4.75, 0.14 + 0.04) {$A_3$};
	\draw (4, 0.75+0.25) rectangle (4.5, 0.75-1); 
	\node at (1.8, 1.5-0.6 + 0.04) {$R_1$};
	\node at (3.75, 1.5-0.6 + 0.04) {$R_2$};
	\node at (4.75, 1.5-0.6 + 0.04) {$R_3$};
	\draw (4.5, 1.5-0.75) -- (5, 1.5-0.75);
	\node at (5.75+1, 1.5-0.6 + 0.04) {$R_n$};
	\node at (4.25, 0.75-0.375) {$\mathcal{S}^2$}; 
	\node at (5.25,0.75-0.375) {$\cdots$};
	\draw (4.5+1, 1.5-0.75) -- (6+1, 1.5-0.75);
	\draw (4.5+1,0) -- (5+1,0);
	\node at (4.75+1, 0.14 + 0.04) {$A_n$};
	\node at (5.75+1, 0.14 + 0.04) {$B_n$};
	\draw (5.5+1, 0) -- (6+1, 0);
	\draw (5+1,-0.25) rectangle (5.5+1,0.25);
	\node at (5.25+1, 0) {$\mathcal{N}^{\bm{\theta}}$};
	\draw (6+1,0.75-1) rectangle (6.75+1,0.75+0.25); 
	\node at (6.375+1,0.75-0.375) {$\Lambda^{\hat{\bm{\theta}}}$};
	\draw[double distance between line centers=0.2em] (7.75, 0.75-0.375) -- (8+0.125, 0.75-0.375);
	\node at (7.92 + 0.0625, 0.75-0.2 + 0.04) {$\hat{\bm{\theta}}$};
\end{tikzpicture}
\caption{Processing $n$ uses of channel $\mathcal{N}^{\bm{\theta}}$ in a sequential or adaptive manner is the most general approach to channel parameter estimation. The $n$ uses of the channel are interleaved with $n-1$ quantum channels $\mathcal{S}^{1}$ through $\mathcal{S}^{n-1}$, which can also share memory systems with each other. The final measurement's outcome is then used to obtain an estimate of the unknown parameter vector~$\bm{\theta}$.  }
\label{fig:sequential strategies}
\end{figure*}
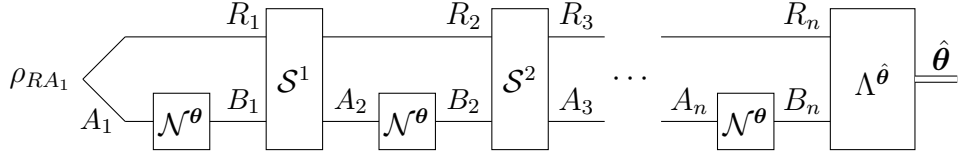

\section{Parameter Estimation}

The goal of classical parameter estimation is to obtain an estimate $\hat{\theta}$ of an unknown real parameter $\theta$ embedded in a probability distribution $p_\theta(x)$, corresponding to a random variable $X$, where we have suppressed the dependence of $X$ on $\theta$ in our notation. That is, we guess the value of $\theta$ from a realization~$x$ of the random variable $X$. The estimate $\hat{\theta}(X)$ is itself a random variable, being a function of $X$. It is common to assume unbiased estimators, i.e., $\mathbb{E} [\hat{\theta}(X)] = \theta$.
The unbiasedness condition means that the estimator is accurate and has no systematic error. Hence we focus on studying the precision of unbiased estimators.

A natural metric to benchmark the performance of an estimator is the mean-squared error $\mathbb{E}[(\hat{\theta}(X) - \theta)^2]$, abbreviated as MSE. For an unbiased estimator, the MSE is equal to the variance $\text{Var} (\hat{\theta}(X))$. The Cramer--Rao bound (CRB) places a lower bound on the variance of an unbiased estimator \cite{Cram46, Rao45, Kay93}: 
\begin{equation}
	\text{Var}(\hat{\theta}(X)) \geq \frac{1}{I_{F}(\theta; \{p_{\theta}\}_\theta)},	
\end{equation}
where $I_F(\theta; \{p_{\theta}\}_\theta)$ is the classical Fisher information, defined as
\begin{equation}
I_{F} (\theta; \{p_{\theta}\}_\theta) \coloneqq \mathbb{E}[(\partial_\theta \ln p_\theta(X))^2].
\label{eq:classical-FI}
\end{equation}
If one has $n$ independent samples $x^n$ described by the random sequence $X^n \equiv X_1, \ldots, X_n$, then the corresponding CRB is
\begin{equation}
	\text{Var}(\hat{\theta}(X^n)) \geq \frac{1}{n I_{F}(\theta; \{p_{\theta}\}_\theta)}.	
\end{equation}

At this stage, we would like to clarify that our approach adopts the frequentist approach to parameter estimation. In general, the MSE and Cramer--Rao bounds may depend on the value of the unknown parameter, in contrast with the more general paradigm of Bayesian parameter estimation \cite{Li2018}. This parameter dependence is alleviated by enforcing the unbiasedness condition.

Estimation theory has been generalized to quantum systems, where quantum probes and quantum measurements are allowed (see, e.g., \cite{Sidhu2019,DS15} for recent reviews). We then have an infinite number of logarithmic derivative operators that each reduce to the logarithmic derivative $\partial_\theta \ln p_\theta(x)$ in the classical case in \eqref{eq:classical-FI}. One such example is the right logarithmic derivative (RLD) operator, defined implicitly via the equation $\partial_{\theta}\rho_{\theta}=\rho_{\theta}R_\theta$,	
where $\{\rho_{\theta}\}_{\theta}$ is a differentiable family of quantum states. The RLD Fisher information will be formally defined in the next section. 


Consider the task of estimating a real parameter $\theta$ encoded in a quantum channel $\mathcal{N}^{\theta}_{A \rightarrow B}$ (see \cite{Sasaki2002, Fujiwara_2003, Fujiwara2004, Ji2008, Fujiwara2008, Mat10, Hayashi2011, Demkowicz-Dobrzanski2012, Kolodynski2013, Demkowicz-Dobrzanski2014, Sekatski2017, Demkowicz-Dobrzanski2017, Zhou2018, Zhou2019, Yang2020, Zhou2020} for  extensive literature on this problem). The most general setting for this problem, given $n$ uses of the channel, consists of performing the estimation via a sequential strategy \cite{Giovannetti2006,PhysRevLett.98.090501,Demkowicz-Dobrzanski2014, Yuan2016, Yuan2017, Zhou2020} as depicted in Figure~\ref{fig:sequential strategies}. Parallel estimation strategies, wherein the $n$ calls to the channel are made simultaneously, are contained within the larger family of sequential strategies. Assuming an unbiased estimator, the following quantum Cramer--Rao bound (QCRB) holds in this general setting of channel estimation \cite{KW20a}: 
\begin{equation}
	\text{Var} (\hat{\theta}) \geq \frac{1}{n \widehat{I}_F ( \theta; \{ \mathcal{N}^{\theta}_{A \rightarrow B}\}_\theta)},
	\label{eq:main-result}
\end{equation}
where $\widehat{I}_F (\theta; \{ \mathcal{N}^{\theta}_{A \rightarrow B}\}_\theta)$ is the RLD Fisher information for the parameter $\theta$ encoded in the channel family $\{ \mathcal{N}^{\theta}_{A \rightarrow B}\}_\theta$ \cite{Hayashi2011}, and we define it formally later in \eqref{eq:rld-qfi-channels}. The result in~\eqref{eq:main-result} is presented in Ref. \cite{KW20a} and is a special case of the more general result reported here in~\eqref{eq:single-letter-multi-param-crb}. A restricted subset of the most general sequential strategies consists of parallel strategies, where the $n$ uses of the channel are made simultaneously using an entangled probe state, and the bound in \eqref{eq:main-result} was previously established in \cite{Hayashi2011} for this special case. In a parallel strategy, the optimal probe state may have an undesired dependence on the unknown parameter $\theta$. This dependence, however, can be circumvented in a sequential strategy due to the ability to perform adaptive control (channels $\mathcal{S}^1$ through $\mathcal{S}^{n-1}$ in Figure~\ref{fig:sequential strategies}) between uses of the unknown channel.


The task of simultaneously estimating multiple parameters is a much more involved task than estimating a single parameter; however, Cramer--Rao bounds can still be constructed. This problem, too, has an extensive literature and a number of important recent results \cite{Hel67, Hol72, YL73, Belavkin1976a, Bagan2006, H11book, Monras2011, Humphreys2013, Yue2015, Ragy2016, Sidhu2019a, Albarelli2019, Tsang2019a, Yang2019c, Albarelli2020, Demkowicz-Dobrzanski2020a, Friel2020, Gorecki2020}. See \cite{Szczykulska2016, Albarelli2020a} for recent reviews on multiparameter estimation. In the quantum case, an additional complication is that the optimal measurements for each parameter may not be compatible. Consider that $D$ parameters need to be estimated and are encoded in a vector $\bm{\theta} \coloneqq [\theta_1 ~ \theta_2 ~ \cdots ~ \theta_D]^T$. For a differentiable family $\{ \rho_{\bm{\theta}} \}_{\bm{\theta}}$ of quantum states, one defines $D$ RLD operators using the equations $\partial_{\theta_{j}} \rho_{\bm{\theta}} = \rho_{\bm{\theta}} R_{\theta_j}$ for $j \in \{ 1, 2, \dots D\}$. The RLD Fisher information, instead of being a scalar, takes the form of a $D \times D$ matrix with elements $\widehat{I}_F(\bm{\theta};\{ \rho_{\bm{\theta}}\}_{\bm{\theta}})_{j,k}\coloneqq \operatorname{Tr}[R_{\theta_{j}}^{\dag}\rho_{\bm{\theta}}R_{\theta_{k}}]$. This leads to the following matrix Cramer--Rao bound \cite{YL73}:
\begin{equation}
	\text{Cov}(\bm{\theta}) \geq \widehat{I}_F(\bm{\theta}; \{\rho_{\bm{\theta}}\}_{\bm{\theta}})^{-1}.
	\label{eq:matrix-CRB-1}
\end{equation}
In the above, $\text{Cov} (\bm{\theta})$ is a covariance matrix with matrix elements defined as 
\begin{equation}
	[\text{Cov}(\bm{\theta})]_{j,k} = \sum_l \Tr[ M_l \rho_{\bm{\theta}} ] (\hat{\theta}_j(l) - \theta_j) (\hat{\theta}_k(l) - \theta_k),
\end{equation}
where $M_l \geq 0$ are measurement operators satisfying $\sum_l M_l = \mathbb{I}$, and
\begin{equation}
\hat{\bm{\theta}}(l) \coloneqq [\hat{\theta}_1(l) ~ \hat{\theta}_2(l) ~ \cdots ~ \hat{\theta}_D(l)]^T
\end{equation}
 is a function that maps the measurement result $l$ to an estimate of the parameters $\bm{\theta}$. 

\section{Right Logarithmic Derivative}

The reason that there are an infinite number of noncommutative Fisher information quantities is that in quantum estimation theory, the logarithmic derivative takes the form of an operator, rather than a scalar, as it does in the classical case. Each noncommutative Fisher information collapses to the scalar Fisher information in the classical case, i.e. when only diagonal density operators are involved. One particular noncommutative generalization arises from the right logarithmic derivative (RLD) operator: For a family of states $\{\rho_{\theta}\}_{\theta}$, the RLD operator $R_\theta$ is implicitly defined by the equation $\partial_{\theta}\rho_{\theta}=\rho_{\theta}R_\theta$.

The RLD operator introduced above leads to one quantum generalization of the classical logarithmic derivative \cite{YL73}. This is the one we focus on in this paper. We begin by defining it for the case of single parameter estimation. Consider a single unknown parameter $\theta$ embedded in a quantum state $\rho_\theta$. The RLD Fisher information of a differentiable family $\{\rho_{\theta}\}_{\theta}$ of states is defined as
\begin{equation}
	\widehat{I}_{F}(\theta;\{\rho_{\theta}\}_{\theta})=
		\operatorname{Tr}[(\partial_{\theta}\rho_{\theta})^{2}\rho_{\theta}^{-1}] ,
\end{equation}
if $\operatorname{supp}(\partial_{\theta}\rho_{\theta})\subseteq \operatorname{supp}(\rho_{\theta})$,
and it is set to $+\infty$ otherwise. The inverse $\rho_{\theta}^{-1}$ is taken on the support of $\rho_{\theta}$. Alternatively, if the support condition $\operatorname{supp}(\partial_{\theta}\rho_{\theta})\subseteq \operatorname{supp}(\rho_{\theta})$ is satisfied, then the RLD Fisher information can also be defined using the RLD operator $R_\theta$ as follows:
\begin{equation}
\widehat{I}_{F}(\theta;\{\rho_{\theta}\}_{\theta}) = \Tr[R_\theta^2 \rho_\theta ].
\end{equation}

The RLD Fisher information is the largest noncommutative Fisher information \cite{Petz2011}.
 This mirrors the fact that the geometric R\'enyi relative entropy is the largest R\'enyi relative entropy that satisfies the data-processing inequality \cite{Mat13}.  
 We have explored both the qualitative and quantitative connections between the RLD Fisher information and geometric R\'enyi relative entropy (and hence, between quantum channel estimation and discrimination) in detail in~\cite{KW20a}.

If, instead of a state family $\{ \rho_\theta \}_\theta$, we have a differentiable channel family $\{ \mathcal{N}^{\theta}_{A \rightarrow B} \}_\theta$, then the RLD Fisher information of this channel family is defined as
\begin{align}
	\widehat{I}_{F}(\theta;\{\mathcal{N}_{A\rightarrow B}^{\theta}\}_{\theta})& \coloneqq
		\sup_{\rho_{RA}}\widehat{I}_{F}(\theta;\{\mathcal{N}_{A\rightarrow B}^{\theta}(\rho_{RA})\}_{\theta})
	\label{eq:RLD-qfi-channels-first-def}\\
	& = \left\Vert \operatorname{Tr}_{B}[(\partial_{\theta}\Gamma_{RB}^{\mathcal{N}^{\theta}})(\Gamma_{RB}^{\mathcal{N}^{\theta}})^{-1}(\partial_{\theta} \Gamma_{RB}^{\mathcal{N}^{\theta}})]\right\Vert _{\infty} \label{eq:rld-qfi-channels},
\end{align}
 if $\operatorname{supp}(\partial_{\theta}\Gamma_{RB}^{\mathcal{N}^{\theta}
})\subseteq\operatorname{supp}(\Gamma_{RB}^{\mathcal{N}^{\theta}})$,
and it is equal to $+\infty$ otherwise, where 
\begin{equation}
\Gamma^{\mathcal{N}^{\theta}}_{RB} \coloneqq
\sum_{i,j} |i\rangle\!\langle j|_R \otimes 
\mathcal{N}^{\theta}_{A\to B}(|i\rangle\!\langle j|_A)
\end{equation} is the Choi operator of the channel~$\mathcal{N}^{\theta}_{A\to B}$. The optimization in \eqref{eq:RLD-qfi-channels-first-def} is with respect to every bipartite input state $\rho_{RA}$ that has no dependence on the parameter $\theta$, and the equality in \eqref{eq:rld-qfi-channels} was established in \cite{Hayashi2011}. Note that the optimal value in \eqref{eq:RLD-qfi-channels-first-def} is achieved by a pure bipartite state with system $R$ isomorphic to system $A$. The RLD Fisher information of a channel family $\widehat{I}_{F}(\theta;\{\mathcal{N}_{A\rightarrow B}^{\theta}\}_{\theta})$ is efficiently computable via a semi-definite program \cite{KW20a}. We also note here that what we define as the RLD Fisher information in \eqref{eq:RLD-qfi-channels-first-def} and \eqref{eq:rld-qfi-channels} is also known as the conditional Fisher information \cite{Heyde1975, Koenig2015, DePalma2017}, since we consider known probe states. 


The RLD Fisher information of isometric or unitary channels is an uninteresting information measure. This is because a differentiable family $\{\mathcal{U}_\theta\}_\theta$ of isometric or unitary channels induces a differentiable family of pure states. The RLD Fisher information in this case is either infinite or zero, neither of which lead to a useful CRB.

In multiparameter estimation, we have a differentiable family $\{ \rho_{\bm{\theta}} \}_{\bm{\theta}}$ of quantum states. Let $\Pi^{\perp}_{\rho_{\bm{\theta}}}$ denote the projection onto the kernel of $\rho_{\bm{\theta}}$. In the case that the following finiteness conditions hold
\begin{equation}
(\partial_{\theta_{j}}\rho_{\bm{\theta}})
(\partial_{\theta_{k}}\rho_{\bm{\theta}})\Pi_{\rho_{\bm{\theta}}}^{\perp}
  =0\quad\forall j,k\in\left\{  1,\ldots,D\right\},
  \label{eq:app:finiteness-RLD-matrix}
\end{equation}
the matrix elements of the multiparameter RLD\ Fisher information matrix are defined as follows:%
\begin{equation}
	[\widehat{I}_F(\bm{\theta};\{\rho_{\bm{\theta}}\}_{\bm{\theta}})]_{j,k}\coloneqq 
	\operatorname{Tr}[(\partial_{\theta_{j}}\rho_{\bm{\theta}})\rho
_{\bm{\theta}}^{-1}(\partial_{\theta_{k}}\rho_{\bm{\theta}})].
\end{equation} 
We also note here that the finiteness conditions in \eqref{eq:app:finiteness-RLD-matrix} can be equivalently written as
\begin{equation} \label{eq:alternate-finiteness-condition}
	(\partial_{\theta_{k}}\rho_{\bm{\theta}})\Pi_{\rho_{\bm{\theta}}}^{\perp}
	=0\quad\forall k\in\left\{  1,\ldots,D\right\}.
\end{equation}
To see that \eqref{eq:app:finiteness-RLD-matrix} and \eqref{eq:alternate-finiteness-condition} are equivalent, suppose in \eqref{eq:app:finiteness-RLD-matrix} that $j=k$. We then have
\begin{equation}
	(\partial_{\theta_{k}}\rho_{\bm{\theta}})
(\partial_{\theta_{k}}\rho_{\bm{\theta}})\Pi_{\rho_{\bm{\theta}}}^{\perp}
  =0\quad\forall k\in\left\{  1,\ldots,D\right\}	
\end{equation}
which implies that
\begin{equation}
	(\partial_{\theta_{k}}\rho_{\bm{\theta}})\Pi_{\rho_{\bm{\theta}}}^{\perp}
	=0\quad\forall k\in\left\{  1,\ldots,D\right\}.	
\end{equation}
Thus we have that \eqref{eq:app:finiteness-RLD-matrix}$\implies$\eqref{eq:alternate-finiteness-condition}. Next, we show that \eqref{eq:alternate-finiteness-condition}$\implies$\eqref{eq:app:finiteness-RLD-matrix}. Suppose for all $k\in\left\{  1,\ldots,D\right\}$ that
\begin{equation}
	(\partial_{\theta_{k}}\rho_{\bm{\theta}})\Pi_{\rho_{\bm{\theta}}}^{\perp} = 0.
\end{equation}
Now we multiply on both sides by $(\partial_{\theta_{j}}\rho_{\bm{\theta}})$ to yield
\begin{equation}
	(\partial_{\theta_{j}}\rho_{\bm{\theta}}) (\partial_{\theta_{k}}\rho_{\bm{\theta}})\Pi_{\rho_{\bm{\theta}}}^{\perp} = 0 \quad\forall j,k\in\left\{  1,\ldots,D\right\}	
\end{equation}
and thus we conclude that \eqref{eq:alternate-finiteness-condition}$\iff$\eqref{eq:app:finiteness-RLD-matrix}.

The RLD Fisher information matrix is then defined as follows:
\begin{align}
\widehat{I}_F(\bm{\theta};\{\rho_{\bm{\theta}}\}_{\bm{\theta}}) &
\coloneqq \sum_{j,k=1}^{D}\operatorname{Tr}[(\partial_{\theta_{j}}\rho_{\bm{\theta }}%
)\rho_{\bm{\theta}}^{-1}(\partial_{\theta_{k}}\rho_{\bm{\theta}}%
)]|j\rangle\!\langle k|\\
&  =\operatorname{Tr}_{2}\!\left[  \sum_{j,k=1}^{D}|j\rangle\!\langle
k|\otimes(\partial_{\theta_{j}}\rho_{\bm{\theta}})\rho_{\bm{\theta}}%
^{-1}(\partial_{\theta_{k}}\rho_{\bm{\theta}})\right]  ,
\end{align}
where $\Tr_2[\dots]$ refers to tracing over the second subsystem. If the finiteness conditions in \eqref{eq:app:finiteness-RLD-matrix} do not hold, then some elements of the RLD Fisher information matrix are infinite.

As stated in the previous section, the RLD Fisher information matrix is featured in the matrix quantum Cramer--Rao inequality in \eqref{eq:matrix-CRB-1}.
To obtain scalar Cramer--Rao bounds from this matrix inequality, we define a $D \times D$ positive semi-definite, unit trace weight matrix $W$ (also known as the risk matrix). The weight matrix $W$ can be chosen in accordance with the goal of the parameter estimation experiment. Particular choices of $W$ can lead to changes in the optimal estimation strategy.
\begin{remark}
Note that, throughout our paper, a weight matrix $W$ is defined to be a positive semi-definite, unit trace matrix. In general, it may have complex entries. 
\end{remark}
The
RLD\ Fisher information value is then defined as%
\begin{align}
\widehat{I}_F(\bm{\theta},W;\{\rho_{\bm{\theta}}\}_{\bm{\theta}}) &
\coloneqq \operatorname{Tr}[W\widehat{I}_F(\bm{\theta};\{\rho_{\bm{\theta}}\}_{\bm{\theta}})] \label{eq:states-rld-fisher-value} \\ 
& = \operatorname{Tr}\!\left[  (W\otimes I_{d})\left(  \sum_{j,k=1}^{D}%
|j\rangle\!\langle k|\otimes(\partial_{\theta_{j}}\rho_{\bm{\theta}}%
)\rho_{\bm{\theta}}^{-1}(\partial_{\theta_{k}}\rho_{\bm{\theta}})\right)
\right]  \\
&  =\sum_{j,k=1}^{D}\langle k|W|j\rangle\operatorname{Tr}\!\left[
(\partial_{\theta_{j}}\rho_{\bm{\theta}})\rho_{\bm{\theta}}^{-1}%
(\partial_{\theta_{k}}\rho_{\bm{\theta}})\right]  .
\end{align}

The finiteness condition for the RLD Fisher information value defined above is
\begin{equation}
\left[  \sum_{j,k=1}^{D}\langle k|W|j\rangle(\partial_{\theta_{k}}%
\rho_{\bm{\theta}})(\partial_{\theta_{j}}\rho_{\bm{\theta}})\right]  \Pi
_{\rho_{\bm{\theta}}}^{\perp}  =0, \label{eq:app:state-RLD-finiteness-cond} 
\end{equation}
where $\Pi_{\rho_{\bm{\theta}}}^{\perp}$ is the projection onto the kernel of $\rho_{\bm{\theta}}$.

\begin{theorem} \label{thm:state-rld-scalar-inequality}
	The following scalar Cramer--Rao bound holds for estimating multiple parameters $\bm{\theta}$ encoded in a family of quantum states $\{ \rho_{\bm{\theta}} \}_{\bm{\theta}}$:
	\begin{equation} \label{eq:state-rld-scalar-inequality}
	\Tr[ W \operatorname{Cov}(\bm{\theta}) ] \geq \frac{1}{   \widehat{I}_F (\bm{\theta},W;\{\rho_{\bm{\theta}}\}_{\bm{\theta}}) },
	\end{equation}
	where the weight matrix $W$ satisfies $\Tr[W]=1$.
\end{theorem}

\begin{proof}
We start with the following matrix inequality \cite{Hel76, Helstrom1973, Sidhu2019}, as recalled in \eqref{eq:matrix-CRB-1}:
\begin{equation} \label{eq:app-matrix-crb}
	\text{Cov}(\bm{\theta}) \geq \widehat{I}_F(\bm{\theta};\{\rho_{\bm{\theta}}\}_{\bm{\theta}})^{-1}.
\end{equation}
Let us take $W$ as a non-zero, positive semi-definite matrix, and then define the normalized operator~$W'$ as
\begin{equation}
	W' \coloneqq \frac{W}{\Tr[W]}.
\end{equation}
The matrix inequality in \eqref{eq:app-matrix-crb} implies the following:
\begin{align}
	\Tr[W] \Tr[W' \text{Cov}(\bm{\theta}) ] &\geq \Tr[W] \Tr[W' \widehat{I}_F(\bm{\theta};\{\rho_{\bm{\theta}}\}_{\bm{\theta}})^{-1}] \\
									  &= \Tr[W] \Tr[W'^{1/2} \widehat{I}_F(\bm{\theta};\{\rho_{\bm{\theta}}\}_{\bm{\theta}})^{-1} W'^{1/2}] \\
& = \Tr[W] \sum_k \langle k | W'^{1/2} \widehat{I}_F(\bm{\theta};\{\rho_{\bm{\theta}}\}_{\bm{\theta}})^{-1} W'^{1/2} | k \rangle
\\									  &\geq 
\Tr[W]  \left[\sum_k \langle k | W'^{1/2} \widehat{I}_F(\bm{\theta};\{\rho_{\bm{\theta}}\}_{\bm{\theta}}) W'^{1/2} | k \rangle\right]^{-1}
\\& = \frac{\Tr[W]}{ \Tr[ W'^{1/2} \widehat{I}_F(\bm{\theta};\{\rho_{\bm{\theta}}\}_{\bm{\theta}}) W'^{1/2} ]  } \\
									  &= \frac{\Tr[W]}{ \Tr[W' \widehat{I}_F(\bm{\theta};\{\rho_{\bm{\theta}}\}_{\bm{\theta}})] } \\
									  &= \frac{\Tr[W]^2}{ \Tr[W \widehat{I}_F(\bm{\theta};\{\rho_{\bm{\theta}}\}_{\bm{\theta}})] }.
\end{align}
The first inequality is a consequence of \eqref{eq:app-matrix-crb}. The first equality follows from cyclicity of trace. The second inequality uses the operator Jensen inequality \cite{HP03} for the operator convex function $f(x) = x^{-1}$. This inequality is saturated if and only if $\widehat{I}_F(\bm{\theta};\{\rho_{\bm{\theta}}\}_{\bm{\theta}})$ is diagonal in the eigenbasis of $W$, i.e., $\left[\widehat{I}_F(\bm{\theta};\{\rho_{\bm{\theta}}\}_{\bm{\theta}}), W  \right] = 0$. The next equality comes again from cyclicity of trace, and the last equality comes from the definition of $W'$.

The reasoning above leads us to
\begin{equation}
	\Tr[W \text{Cov} (\bm{\theta})] \geq \frac{(\Tr[W])^2}{\Tr[W \widehat{I}_F(\bm{\theta};\{\rho_{\bm{\theta}}\}_{\bm{\theta}})]}
\end{equation}
for every positive semi-definite matrix $W$, 
which implies that
\begin{equation} 
	\Tr[W' \text{Cov} (\bm{\theta})] \geq \frac{1}{\Tr[W' \widehat{I}_F(\bm{\theta};\{\rho_{\bm{\theta}}\}_{\bm{\theta}})]}
\end{equation}
for every positive semi-definite matrix $W'$ such that $\Tr[W'] = 1$.
\end{proof}

We note here that an alternate scalar Cramer--Rao bound involving the RLD Fisher information is the following \cite{YL73, H11book}:
\begin{equation} \label{eq:nagaoka-rld-bound}
		\Tr[ W \operatorname{Cov}(\bm{\theta}) ] \geq \Tr\!\left[ W \text{Re}(\widehat{I}_F(\bm{\theta};\{\rho_{\bm{\theta}}\}_{\bm{\theta}})^{-1}) \right] + \Tr\!\left[ | W^{1/2} \text{Im} (\widehat{I}_F(\bm{\theta};\{\rho_{\bm{\theta}}\}_{\bm{\theta}})^{-1}) W^{1/2} | \right].
\end{equation}
For a matrix $X$, here we define
\begin{equation}
\text{Re}(X) \coloneqq \frac{1}{2} (X + \bar{X}),
\qquad
\text{Im}(X) \coloneqq \frac{1}{2i}(X - \bar{X}),
\end{equation}
where $\bar{X}$ is the matrix with entries that are the complex conjugates of the entries of $X$. This bound, similar to our bound in \eqref{eq:state-rld-scalar-inequality}, arises from the matrix inequality \eqref{eq:app-matrix-crb}. Our bound in \eqref{eq:state-rld-scalar-inequality} is, in general, looser than \eqref{eq:nagaoka-rld-bound}. 
Our bound in \eqref{eq:state-rld-scalar-inequality} is explicitly written as
\begin{equation}
	\Tr[W \widehat{I}_F(\bm{\theta};\{\rho_{\bm{\theta}}\}_{\bm{\theta}})] = \Tr\!\left[W \text{Re} \! \left(\widehat{I}_F(\bm{\theta};\{\rho_{\bm{\theta}}\}_{\bm{\theta}})\right)\right] + i \Tr\!\left[W \text{Im} \! \left(\widehat{I}_F(\bm{\theta};\{\rho_{\bm{\theta}}\}_{\bm{\theta}})\right)\right],
\end{equation}
wherein it is clear that the major difference between it and \eqref{eq:nagaoka-rld-bound} is the absence of any matrix inverses.
Our bound's relative looseness arises because of the use of the operator Jensen inequality (whose details are in the proof of Theorem \ref{thm:state-rld-scalar-inequality}). 
Despite this, our approach has the advantage that it can incorporate the chain rule and amortization collapse for the RLD Fisher information, which are discussed in detail in Section \ref{sec:amortized-fisher-information}. This is directly responsible for the single-letter nature of our bound in \eqref{eq:single-letter-multi-param-crb} and its applicability to general sequential estimation strategies. For state estimation tasks in which the statistical model obeys a condition called $D$-invariance \cite{H11book}, we note here that the RLD-based bound \eqref{eq:nagaoka-rld-bound} is both achievable and the most informative Cramer--Rao bound \cite{Suzuki2016}. 

The RLD Fisher information value of quantum states, as defined in \eqref{eq:states-rld-fisher-value}, is computable via a semi-definite program.

\begin{proposition}
	Let $\{\rho_{A}^{\bm{\theta}}\}_{\bm{\theta}}$ be a differentiable family of quantum states, and let $W$ be a $D\times D$ weight matrix.
	Suppose that \eqref{eq:app:state-RLD-finiteness-cond} holds.
	Then the RLD Fisher information value of quantum states can be calculated via the following semi-definite program:
	\begin{multline} \label{eq:primal-rld-value-states}
	\widehat{I}_{F}(\bm{\theta}, W; \{\rho_{A}^{\bm{\theta}}\})=\inf\Big\{  \operatorname{Tr}
	[(W_F \otimes I_A)M_{FA}]:M_{FA}\geq0, \\
	\begin{bmatrix}
	M_{FA} & \sum_{j =1}^D \ket{j}_F \otimes (\partial_{\theta_j} \rho^{{\bm{\theta}}}_{A} )\\
	\sum_{j =1}^D \bra{j}_F \otimes (\partial_{\theta_j} \rho^{{\bm{\theta}}}_{A} ) & \rho_A^{\bm{\theta}}
	\end{bmatrix}
	\geq0\Big\}  .
	\end{multline}
	
	The dual program is given by
	\begin{equation} \label{eq:dual-rld-value-states}
	\sup_{P_{FA}, Q_{FA\to A}, R_A} 2 \left(\sum_{j=1}^D \operatorname{Re}[\operatorname{Tr}
	[ Q_{FA\to A} ( | j \rangle_F \otimes  (\partial_{\theta_j} \rho_A^{\bm{\theta}})) ]] \right)-\operatorname{Tr}[R_A \rho_A^{\bm{\theta}}],
	\end{equation}
	subject to
	\begin{equation}
	P_{FA} \leq (W_F \otimes I_A),\quad
	\begin{bmatrix}
	P_{FA} & (Q_{FA\to A})^{\dag}\\
	Q_{FA\to A} & R_A
	\end{bmatrix}
	\geq0,
	\end{equation}
	where $P_{FA}$ and $R_A$ are Hermitian, and $Q_{FA\to A}$ is a linear operator.
\end{proposition}

\begin{proof}
	We begin with the formula
	\begin{align}
	\widehat{I}(\bm{\theta},W;\{\rho^{\bm{\theta}}_A\}_{\bm{\theta}}) &
	:=\operatorname{Tr}\!\left[  (W_F \otimes I_{A})\left(  \sum_{j,k=1}^{D}%
	|j\rangle\!\langle k|_F \otimes(\partial_{\theta_{j}}\rho_A^{\bm{\theta}}%
	)(\rho_A^{\bm{\theta}})^{-1}(\partial_{\theta_{k}}\rho_A^{\bm{\theta}})\right)
	\right]  \\
	&  =\sum_{j,k=1}^{D}\langle k|W|j\rangle\operatorname{Tr}\!\left[
	(\partial_{\theta_{j}}\rho_A^{\bm{\theta}})(\rho_A^{\bm{\theta}})^{-1}%
	(\partial_{\theta_{k}}\rho_A^{\bm{\theta}})\right]  .
	\end{align}
	The above
	can be written as
	\begin{equation}
	\Tr \left[ (W_F \otimes I_{A}) (X^{\dag} Y^{-1} X) \right]
	\end{equation}
	where $X = \sum_{j =1}^D \bra{j}_F \otimes (\partial_{\theta_j} \rho^{{\bm{\theta}}}_{A} )$ and $Y = \rho_A^{\bm{\theta}}$.
	
	We next use the Schur complement lemma,
	\begin{equation}
	\begin{bmatrix}
	M & X^{\dag}\\
	X & Y
	\end{bmatrix}
	\geq0\qquad\Longleftrightarrow\qquad Y\geq0,\quad M\geq X^{\dag}Y^{-1}X
	\end{equation}
	which lets us write 
	\begin{equation}
	X^{\dag}Y^{-1}X=\min\left\{  M:
	\begin{bmatrix}
	M & X^{\dag}\\
	X & Y
	\end{bmatrix}
	\geq0\right\},
	\end{equation}
	where the ordering is understood in the L\"owner sense.
	
	Combining the above, we obtain the desired primal form in \eqref{eq:primal-rld-value-states}.
	
	To obtain the dual program in \eqref{eq:dual-rld-value-states}, we apply Lemma \ref{lemma:dual-form}, which we provide in Appendix~\ref{app:lemma-dual-form}.
\end{proof}

\subsection{RLD Fisher Information Value of Quantum Channels}

The RLD Fisher information value for a differentiable family $\{\mathcal{N}_{A\rightarrow
B}^{\bm{\theta}}\}_{\bm{\theta}}$ of quantum channels is defined as%
\begin{equation}
\widehat{I}_F(\bm{\theta},W;\{\mathcal{N}_{A\rightarrow B}^{\bm{\theta}}%
\}_{\bm{\theta}})\coloneqq \sup_{\rho_{RA}}\widehat{I}_F(\bm{\theta },W;\{\mathcal{N}%
_{A\rightarrow B}^{\bm{\theta}}(\rho_{RA})\}_{\bm{\theta}}),
\end{equation}
where the optimization is with respect to every bipartite state $\rho_{RA}$ with system $R$ arbitrarily large. 
However, note that, by a standard argument, it suffices to optimize over pure states $\psi_{RA}$ with system $R$ isomorphic to the channel input system $A$.

The finiteness condition for the channel RLD Fisher information value is

\begin{equation}
\left[  \sum_{j,k=1}^{D}\langle k|W|j\rangle(\partial_{\theta_{k}}\Gamma
_{RB}^{\mathcal{N}^{\bm{\theta}}})(\partial_{\theta_{j}}\Gamma_{RB}%
^{\mathcal{N}^{\bm{\theta}}})\right]  \Pi_{\Gamma^{\mathcal{N}^{\bm{\theta}}}%
}^{\perp}  =0,
\label{eq:app:channel-RLD-finiteness-cond}
\end{equation}
where $\Pi_{\Gamma^{\mathcal{N}^{\bm{\theta}}}}^{\perp}$ is the projection onto the kernel of $\Gamma^{\mathcal{N}^{\bm{\theta}}}_{RB}$, with $\Gamma^{\mathcal{N}^{\bm{\theta}}}_{RB}$ the Choi operator of the channel $\mathcal{N}^{\bm{\theta}}_{A\to B}$.

\begin{proposition}
\label{prop:geo-fish-explicit-formula-1st-order}Let $\{\mathcal{N}%
_{A\rightarrow B}^{\bm{\theta}}\}_{\bm{\theta}}$ be a differentiable family of
quantum channels, and let $W$ be a $D\times D$ weight matrix. Suppose that \eqref{eq:app:channel-RLD-finiteness-cond} holds. Then the 
RLD\ Fisher information value of quantum channels has the following explicit form:%
\begin{equation}
\widehat{I}_{F}(\bm{\theta},W;\{\mathcal{N}_{A\rightarrow B}^{\bm{\theta}}%
\}_{\bm{\theta}})=\left\Vert \sum_{j,k=1}^{D}\langle k|W|j\rangle
\operatorname{Tr}_{B}[(\partial_{\theta_{j}}\Gamma_{RB}^{\mathcal{N}%
^{\bm{\theta}}})(\Gamma_{RB}^{\mathcal{N}^{\bm{\theta}}})^{-1}(\partial
_{\theta_{k}}\Gamma_{RB}^{\mathcal{N}^{\bm{\theta}}})]\right\Vert _{\infty}.
\label{eq:RLD-value-channels-inf-norm}
\end{equation}
\end{proposition}

\begin{proof}
Recall that every pure state $\psi_{RA}$ can be written as%
\begin{equation}
\psi_{RA}=Z_{R}\Gamma_{RA}Z_{R}^{\dag},
\end{equation}
where $Z_{R}$ is a square operator satisfying $\operatorname{Tr}[Z_{R}^{\dag
}Z_{R}]=1$. This implies that%
\begin{align}
\mathcal{N}_{A\rightarrow B}^{\bm{\theta}}(\psi_{RA}) &  =\mathcal{N}%
_{A\rightarrow B}^{\bm{\theta}}(Z_{R}\Gamma_{RA}Z_{R}^{\dag})\\
&  =Z_{R}\mathcal{N}_{A\rightarrow B}^{\bm{\theta}}(\Gamma_{RA})Z_{R}^{\dag}\\
&  =Z_{R}\Gamma_{RB}^{\mathcal{N}^{\bm{\theta}}}Z_{R}^{\dag}.
\end{align}
It suffices to optimize over pure states $\psi_{RA}$ such that $\psi_{A}>0$
because these states are dense in the set of all pure bipartite states. Then
consider that%
\begin{align}
&  \sup_{\psi_{RA}}\widehat{I}_{F}(\bm{\theta},W;\{\mathcal{N}_{A\rightarrow
B}^{\bm{\theta}}(\psi_{RA})\}_{\bm{\theta}})\nonumber\\
&  =\sup_{\psi_{RA}}\sum_{j,k=1}^{D}\langle k|W|j\rangle\operatorname{Tr}%
\left[  (\partial_{\theta_{j}}\mathcal{N}_{A\rightarrow B}^{\bm{\theta}}%
(\psi_{RA}))(\mathcal{N}_{A\rightarrow B}^{\bm{\theta}}(\psi_{RA}%
))^{-1}(\partial_{\theta_{k}}\mathcal{N}_{A\rightarrow B}^{\bm{\theta}}%
(\psi_{RA}))\right]  \\
&  =\sup_{Z_{R}:\operatorname{Tr}[Z_{R}^{\dag}Z_{R}]=1}\sum_{j,k=1}^{D}\langle
k|W|j\rangle\operatorname{Tr}\!\left[  (\partial_{\theta_{j}}Z_{R}\Gamma
_{RB}^{\mathcal{N}^{\bm{\theta}}}Z_{R}^{\dag})(Z_{R}\Gamma_{RB}^{\mathcal{N}%
^{\bm{\theta}}}Z_{R}^{\dag})^{-1}(\partial_{\theta_{k}}Z_{R}\Gamma
_{RB}^{\mathcal{N}^{\bm{\theta}}}Z_{R}^{\dag})\right]  \\
&  =\sup_{Z_{R}:\operatorname{Tr}[Z_{R}^{\dag}Z_{R}]=1}\sum_{j,k=1}^{D}\langle
k|W|j\rangle\operatorname{Tr}\!\left[  Z_{R}(\partial_{\theta_{j}}\Gamma
_{RB}^{\mathcal{N}^{\bm{\theta}}})Z_{R}^{\dag}Z_{R}^{-\dag}(\Gamma
_{RB}^{\mathcal{N}^{\bm{\theta}}})^{-1}Z_{R}^{-1}Z_{R}(\partial_{\theta_{k}%
}\Gamma_{RB}^{\mathcal{N}^{\bm{\theta}}})Z_{R}^{\dag}\right]  \\
&  =\sup_{Z_{R}:\operatorname{Tr}[Z_{R}^{\dag}Z_{R}]=1}\sum_{j,k=1}^{D}\langle
k|W|j\rangle\operatorname{Tr}\!\left[  Z_{R}^{\dag}Z_{R}(\partial_{\theta_{j}%
}\Gamma_{RB}^{\mathcal{N}^{\bm{\theta}}})(\Gamma_{RB}^{\mathcal{N}%
^{\bm{\theta}}})^{-1}(\partial_{\theta_{k}}\Gamma_{RB}^{\mathcal{N}%
^{\bm{\theta}}})\right]  \\
&  =\sup_{Z_{R}:\operatorname{Tr}[Z_{R}^{\dag}Z_{R}]=1}\operatorname{Tr}\!
\left[  Z_{R}^{\dag}Z_{R}\sum_{j,k=1}^{D}\langle k|W|j\rangle\operatorname{Tr}%
_{B}\!\left[  (\partial_{\theta_{j}}\Gamma_{RB}^{\mathcal{N}^{\bm{\theta}}%
})(\Gamma_{RB}^{\mathcal{N}^{\bm{\theta}}})^{-1}(\partial_{\theta_{k}}%
\Gamma_{RB}^{\mathcal{N}^{\bm{\theta}}})\right]  \right]  \\
&  =\left\Vert \sum_{j,k=1}^{D}\langle k|W|j\rangle\operatorname{Tr}%
_{B}\!\left[  (\partial_{\theta_{j}}\Gamma_{RB}^{\mathcal{N}^{\bm{\theta}}%
})(\Gamma_{RB}^{\mathcal{N}^{\bm{\theta}}})^{-1}(\partial_{\theta_{k}}%
\Gamma_{RB}^{\mathcal{N}^{\bm{\theta}}})\right]  \right\Vert _{\infty}%
\end{align}
The last equality is a consequence of the characterization of the infinity
norm of a positive semi-definite operator $Y$ as $\left\Vert Y\right\Vert
_{\infty}=\sup_{\rho>0,\operatorname{Tr}[\rho]=1}\operatorname{Tr}[Y\rho]$.
\end{proof}

\bigskip

As in the case of the RLD Fisher information value for quantum states, the analogous quantity for channels can also be computed via a semi-definite program.

\begin{proposition} \label{prop:rld-channels-sdp}
	Let $\{\mathcal{N}_{A\rightarrow B}^{\bm{\theta}}\}_{\bm{\theta}}$ be a differentiable family of quantum channels, and let $W$ be a $D\times D$ weight matrix. Suppose that \eqref{eq:app:channel-RLD-finiteness-cond} holds. Then the RLD\ Fisher information value of quantum channels can be calculated via the following semi-definite program:
	\begin{equation} \label{eq:primal-rld-value-channels}
	\widehat{I}_{F}(\bm{\theta},W;\{\mathcal{N}_{A\rightarrow B}^{\bm{\theta}}%
	\}_{\bm{\theta}})=\inf\lambda\in\mathbb{R}^{+}, 
	\end{equation}
	subject to
	\begin{equation} \label{eq:primal-rld-value-channels-2}
	\lambda I_{R}\geq\operatorname{Tr}_{FB}[(W_F \otimes I_{RB} ) M_{FRB}], \qquad
	\begin{bmatrix}
	M_{FRB} & \sum_j | j \rangle_F \otimes  (\partial_{\theta_j} \Gamma^{\mathcal{N}^{\bm{\theta}}}_{RB} )    \\
	\sum_j  \langle j |_F \otimes (\partial_{\theta_j} \Gamma^{\mathcal{N}^{\bm{\theta}}}_{RB} )  & \Gamma^{\mathcal{N}^{\bm{\theta}}}_{RB}
	\end{bmatrix}
	\geq0. 
	\end{equation}
	
	The dual program is given by
	\begin{equation}
	\sup_{\rho_{R}\geq0,P_{FRB},Z_{FRB\to RB},Q_{RB}} 2 \left( \sum_{j=1}^D \operatorname{Re}[\operatorname{Tr}
	[Z_{FRB\to RB} ( | j \rangle_F \otimes  (\partial_{\theta_j} \Gamma^{\mathcal{N}^{\bm{\theta}}}_{RB} )) ]] \right) -\operatorname{Tr}[Q_{RB}\Gamma_{RB}^{\mathcal{N}^{\bm{\theta}}}],
	\end{equation}
	subject to
	\begin{equation}
	\operatorname{Tr}[\rho_{R}]\leq1,\quad
	\begin{bmatrix}
	P_{FRB} & (Z_{FRB\to RB})^{\dag}\\
	Z_{FRB\to RB} & Q_{RB}
	\end{bmatrix}
	\geq0,\quad P_{FRB}\leq\rho_{R} \otimes W_F \otimes I_{B}.
	\end{equation}
		
\end{proposition}

\begin{proof}
	The form of the primal program relies on the combination of a few facts. First, we use the following characterization of the infinity norm of a positive semi-definite operator $A$:
	\begin{equation}
	\left\Vert A\right\Vert _{\infty}=\inf\left\{  \lambda\geq0:A\leq\lambda I\right\}  .
	\end{equation}	
	
	Next we observe that 
	\begin{equation}
	\sum_{j,k=1}^{D}\langle k|W|j\rangle \operatorname{Tr}_{B}[(\partial_{\theta_{j}}\Gamma_{RB}^{\mathcal{N}^{\bm{\theta}}})(\Gamma_{RB}^{\mathcal{N}^{\bm{\theta}}})^{-1}(\partial_{\theta_{k}}\Gamma_{RB}^{\mathcal{N}^{\bm{\theta}}})]
	\end{equation}
	can be written as
	\begin{equation}
	\Tr_{FB} \!\left[ (W_F \otimes I_{RB}) (X^{\dag} Y^{-1} X) \right]
	\end{equation}
	where $X = \sum_{j =1}^D \bra{j}_F \otimes (\partial_{\theta_j} \Gamma^{\mathcal{N}^{\bm{\theta}}}_{RB} )$ and $Y = \Gamma^{\mathcal{N}^{\bm{\theta}}}_{RB}$.
	
	We next use the Schur complement lemma,
	\begin{equation}
	\begin{bmatrix}
	M & X^{\dag}\\
	X & Y
	\end{bmatrix}
	\geq0\qquad\Longleftrightarrow\qquad Y\geq0,\quad M\geq X^{\dag}Y^{-1}X
	\end{equation}
	which lets us write 
	\begin{equation}
	X^{\dag}Y^{-1}X=\min\left\{  M:
	\begin{bmatrix}
	M & X^{\dag}\\
	X & Y
	\end{bmatrix}
	\geq0\right\}
	\end{equation}
	
	Combining the above with \eqref{eq:RLD-value-channels-inf-norm}, we obtain the desired primal form in \eqref{eq:primal-rld-value-channels}.
	
	To arrive at the dual program, we use the standard forms of primal and dual
	semi-definite programs for Hermitian operators $A$ and $B$ and a
	Hermiticity-preserving map $\Phi$ \cite{Wat18}:
	\begin{equation}
	\sup_{X\geq0}\left\{  \operatorname{Tr}[AX]:\Phi(X)\leq B\right\}  ,
	\qquad\inf_{Y\geq0}\left\{  \operatorname{Tr}[BY]:\Phi^{\dag}(Y)\geq
	A\right\}  . \label{eq:standard-SDP-form-RLD-ch-helper-channels}
	\end{equation}
	From \eqref{eq:primal-rld-value-channels}--\eqref{eq:primal-rld-value-channels-2}, we
	identify
	\begin{align}
		B  &  =
		\begin{bmatrix}
			1 & 0\\
			0 & 0
		\end{bmatrix},
		\quad Y=
		\begin{bmatrix}
			\lambda & 0\\
			0 & M_{FRB}
		\end{bmatrix}
		,\quad\Phi^{\dag}(Y)=
		\begin{bmatrix}
			\lambda I_{R}-\operatorname{Tr}_{FB}[(W_F \otimes I_{RB})M_{FRB}] & 0 & 0\\
			0 & M_{FRB} & 0\\
			0 & 0 & 0
		\end{bmatrix}
		,\\
		A  &  =
		\begin{bmatrix}
			0 & 0 & 0\\
			0 & 0 & -\sum_{j=1}^D | j \rangle_F \otimes  (\partial_{\theta_j} \Gamma^{\mathcal{N}^{\bm{\theta}}}_{RB} )\\
			0 & -\sum_{j=1}^{D} \bra{j}_F \otimes  (\partial_{\theta_j} \Gamma^{\mathcal{N}^{\bm{\theta}}}_{RB} ) & -\Gamma^{\mathcal{N}^{\bm{\theta}}}_{RB}
		\end{bmatrix}
		.	
	\end{align}

	Upon setting
	\begin{equation}
	X=
	\begin{bmatrix}
	\rho_{R} & 0 & 0\\
	0 & P_{FRB} & (Z_{FRB\to RB})^{\dag}\\
	0 & Z_{FRB\to RB} & Q_{RB}
	\end{bmatrix}
	,
	\end{equation}
	we find that
	\begin{align}
		& \operatorname{Tr}[X\Phi^{\dag}(Y)]  \notag \\
		&  =\operatorname{Tr}\!\left[
		\begin{bmatrix}
			\rho_{R} & 0 & 0\\
			0 & P_{FRB} & (Z_{FRB\to RB})^{\dag}\\
			0 & Z_{FRB\to RB} & Q_{RB}
		\end{bmatrix}
		\begin{bmatrix}
			\lambda I_{R}-\operatorname{Tr}_{FB}[(W_F \otimes I_{RB})M_{FRB}] & 0 & 0\\
			0 & M_{FRB} & 0\\
			0 & 0 & 0
		\end{bmatrix}
		\right] \\
		&  =\operatorname{Tr}[\rho_{R}(\lambda I_{R}-\operatorname{Tr}_{FB}
		[(W_F \otimes I_{RB})M_{FRB}])]+\operatorname{Tr}[P_{FRB}M_{FRB}]\\
		&  =\lambda\operatorname{Tr}[\rho_{R}]+\operatorname{Tr}[(P_{FRB}-\rho
		_{R}\otimes W_F \otimes I_{B})M_{FRB}]\\
		&  =\operatorname{Tr}\!\left[
		\begin{bmatrix}
			\lambda & 0\\
			0 & M_{FRB}
		\end{bmatrix}
		\begin{bmatrix}
			\operatorname{Tr}[\rho_{R}] & 0\\
			0 & P_{FRB}-\rho_{R} \otimes W_F \otimes I_{B}
		\end{bmatrix}
		\right]  \\
		&= \Tr[Y \Phi(X)]
		,
	\end{align}
	to find that the dual is given by
	\begin{equation}
	\sup_{\substack{\rho_{R},P_{FRB},\\Z_{FRB\to RB},Q_{RB}}}\operatorname{Tr}\left[AX
	\right]  ,
	\end{equation}
	subject to
	\begin{equation}
	\begin{bmatrix}
	\rho_{R} & 0 & 0\\
	0 & P_{FRB} & (Z_{FRB\to RB})^{\dag}\\
	0 & Z_{FRB\to RB} & Q_{RB}
	\end{bmatrix}
	\geq0,\qquad
	\begin{bmatrix}
	\operatorname{Tr}[\rho_{R}] & 0\\
	0 & P_{FRB}-\rho_{R} \otimes W_F \otimes I_{B}
	\end{bmatrix}
	\leq
	\begin{bmatrix}
	1 & 0\\
	0 & 0
	\end{bmatrix}
	.
	\end{equation}
	We can swap $Z_{FRB\to RB} \rightarrow - Z_{FRB\to RB}$ with no change to the optimal value. This leads to the following simplified form of the dual program:
	\begin{equation}
	\sup_{\rho_{R}\geq0,P_{FRB},Z_{FRB},Q_{RB}}2 \sum_{j=1}^D \operatorname{Re}[\operatorname{Tr}
	[Z_{FRB\to RB} ( | j \rangle_F \otimes  \partial_{\theta_j} \Gamma^{\mathcal{N}^{\bm{\theta}}}_{RB}  ) ]]-\operatorname{Tr}[Q_{RB}\Gamma_{RB}^{\mathcal{N}^{\bm{\theta}}}],
	\end{equation}
	subject to
	\begin{equation}
	\operatorname{Tr}[\rho_{R}]\leq1,\quad
	\begin{bmatrix}
	P_{FRB} & -(Z_{FRB\to RB})^{\dag}\\
	-Z_{FRB\to RB} & Q_{RB}
	\end{bmatrix}
	\geq0,\quad P_{FRB}\leq\rho_{R} \otimes W_F \otimes I_{B}.
	\end{equation}
	Then we note that
	\begin{equation}
	\begin{bmatrix}
	P_{FRB} & -(Z_{FRB\to RB})^{\dag}\\
	-Z_{FRB\to RB} & Q_{RB}
	\end{bmatrix}
	\geq0 \quad\Longleftrightarrow\quad
	\begin{bmatrix}
	P_{FRB} & (Z_{FRB\to RB})^{\dag}\\
	Z_{FRB\to RB} & Q_{RB}
	\end{bmatrix}
	\geq0
	\end{equation}
	This concludes the proof.
\end{proof}

\section{Amortized Fisher Information} \label{sec:amortized-fisher-information}

Amortized channel divergences were defined in \cite{Berta2018c} to provide a mathematical framework for studying the power of sequential strategies over parallel ones in quantum channel discrimination. With the view of performing a similar comparison for quantum channel estimation, we define the amortized Fisher information for quantum channel families \cite{KW20a}. For the channel family $\{ \mathcal{N}_{A \rightarrow B}^{\theta} \}_{\theta}$, the amortized RLD Fisher information is defined as

\begin{equation}
	\widehat{I}_{F}^{\mathcal{A}}(\theta;\{\mathcal{N}_{A\rightarrow B}^{\theta}\}_{\theta})\coloneqq \sup_{\{\rho_{RA}^{\theta}\}_{\theta}}\left[  \widehat{I}_{F}(\theta;\{\mathcal{N}_{A\rightarrow B}^{\theta}(\rho_{RA}^{\theta})\}_{\theta})-\widehat{I}_{F}(\theta;\{\rho_{RA}^{\theta}\}_{\theta})\right].
	\label{eq:amortized-RLD-Fisher}
\end{equation}
For the case of multiparameter estimation, we define the amortized RLD Fisher information value as follows:
\begin{equation}
	\widehat{I}_{F}^{\mathcal{A}}(\bm{\theta}, W; \{\mathcal{N}_{A\rightarrow B}^{\bm{\theta}}\}_{\bm{\theta}})\coloneqq  \\ \sup_{\{\rho_{RA}^{\bm{\theta}}\}_{\bm{\theta}}}  \widehat{I}_{F}(\bm{\theta}, W; \{\mathcal{N}_{A\rightarrow B}^{\bm{\theta}}(\rho_{RA}^{\bm{\theta}})\}_{\bm{\theta}})-\widehat{I}_{F}(\bm{\theta}, W; \{\rho_{RA}^{\bm{\theta}}\}_{\bm{\theta}}).
\end{equation}

The framework of amortized Fisher information can be applied more generally beyond RLD Fisher information, as discussed in \cite{KW20a}, and it is helpful for assessing the power of sequential strategies when estimating a parameter encoded in a quantum channel. In the sequential strategies described in Figure~\ref{fig:sequential strategies}, one can qualitatively say that the goal is to ``accumulate'' as much information about $\bm{\theta}$ into the state being carried forward from one channel use to another. The amortized Fisher information captures the marginal increase in Fisher information per channel use in such a scenario. 

A basic inequality obeyed by the amortized RLD Fisher information, as a direct consequence of definitions in \eqref{eq:RLD-qfi-channels-first-def} and \eqref{eq:amortized-RLD-Fisher}, is that 
\begin{equation} \label{eq:amortization-basic-inequality}
	\widehat{I}_{F}^{\mathcal{A}}(\theta;\{\mathcal{N}_{A\rightarrow B}^{\theta}\}_{\theta})\geq\widehat{I}_{F}(\theta;\{\mathcal{N}_{A\rightarrow B}^{\theta}\}_{\theta}),
\end{equation}
which also holds for the multiparameter case as follows:
\begin{equation}
\widehat{I}_{F}^{\mathcal{A}}(\bm{\theta}, W; \{\mathcal{N}_{A\rightarrow B}^{\bm{\theta}}\}_{\bm{\theta}}) \geq \widehat{I}_{F} ( \bm{\theta}, W; \{ \mathcal{N}_{A \rightarrow B}^{\bm{\theta}}\}_{\bm{\theta}} ) .
\end{equation}
This can be understood by considering the right-hand side to arise from restricting to input states with no parameter dependence. Qualitatively, it means that the marginal increase in Fisher information value can only be improved by using a catalyst state family as input. That is,
\begin{align}
\widehat{I}_F^{\mathcal{A}} (\bm{\theta}, W; \{ \mathcal{N}^{\bm{\theta}, W}_{A \rightarrow B} \}_{\bm{\theta}}) &\coloneqq \sup_{\{\rho_{RA}^{\bm{\theta}}\}_{\bm{\theta}}}\left[  \widehat{I}_{F}(\bm{\theta}, W;\{\mathcal{N}_{A\rightarrow B}^{\bm{\theta}}(\rho_{RA}^{\bm{\theta}})\}_{\bm{\theta}})-\widehat{I}_{F}(\bm{\theta}, W;\{\rho_{RA}^{\bm{\theta}}\}_{\bm{\theta}})\right] \\
&\geq \sup_{\{\rho_{RA}\}_{\bm{\theta}}}\left[  \widehat{I}_{F}(\bm{\theta}, W;\{\mathcal{N}_{A\rightarrow B}^{\bm{\theta}}(\rho_{RA})\}_{\bm{\theta}})-\widehat{I}_{F}(\bm{\theta}, W;\{\rho_{RA}\})\right] \\
&=  \sup_{\{\rho_{RA}\}_{\bm{\theta}}}\left[  \widehat{I}_{F}(\bm{\theta}, W;\{\mathcal{N}_{A\rightarrow B}^{\bm{\theta}}(\rho_{RA})\}_{\bm{\theta}})\right] \\
&= \widehat{I}_{F}(\bm{\theta}, W;\{\mathcal{N}_{A\rightarrow B}^{\bm{\theta}}\}_{\bm{\theta}}).
\end{align}

For some amortized channel quantities, the reverse inequality also holds. In fact, we prove in our paper that it holds for the RLD Fisher information value of quantum channels, i.e.,
\begin{equation}
\widehat{I}_{F}^{\mathcal{A}}(\bm{\theta}, W; \{\mathcal{N}_{A\rightarrow B}^{\bm{\theta}}\}_{\bm{\theta}}) \leq \widehat{I}_{F} ( \bm{\theta}, W; \{ \mathcal{N}_{A \rightarrow B}^{\bm{\theta}}\}_{\bm{\theta}} ) . \label{eq:amortization-reverse-inequality}
\end{equation}
This arises as a direct consequence of the chain rule for the RLD Fisher information value, which we state and prove in Proposition \ref{prop:rld-chain-rule}.

A consequence of the reverse inequality \eqref{eq:amortization-reverse-inequality} is a so-called amortization collapse, wherein,
\begin{align}
	\widehat{I}_{F}^{\mathcal{A}}(\theta;\{\mathcal{N}_{A\rightarrow B}^{\theta}\}_{\theta}) & = \widehat{I}_{F}(\theta;\{\mathcal{N}_{A\rightarrow B}^{\theta}\}_{\theta}), \text{~and}
	\label{eq:amort-collapse-1}\\
	\widehat{I}_{F}^{\mathcal{A}}(\bm{\theta}, W; \{\mathcal{N}_{A\rightarrow B}^{\bm{\theta}}\}_{\bm{\theta}}) & = \widehat{I}_{F} ( \bm{\theta}, W; \{ \mathcal{N}_{A \rightarrow B}^{\bm{\theta}}\}_{\bm{\theta}} ) .
	\label{eq:amort-collapse-2}
\end{align}
The meaning of an amortization collapse is that the RLD Fisher information cannot be increased by using a catalyst. We first made the simple observation using \eqref{eq:amortization-basic-inequality} that a catalyst cannot decrease the RLD Fisher information. The reverse inequality \eqref{eq:amortization-reverse-inequality} then implies the opposite, which is that one can do no better with a catalyst state than with one. That is, a catalyst state family $\{ \rho_{RA}^{\bm{\theta}} \}_{\bm{\theta}}$ at the channel input does not increase the RLD Fisher information $\widehat{I}_{F} ( \bm{\theta}, W; \{ \mathcal{N}_{A \rightarrow B}^{\bm{\theta}}\}_{\bm{\theta}})$ by any more than its own RLD Fisher information $\widehat{I}_{F}(\bm{\theta}, W; \{\rho_{RA}^{\bm{\theta}}\}_{\bm{\theta}})$. Continuing the earlier discussion of amortized Fisher information capturing marginal increment in resource, the  amortization collapse in \eqref{eq:amort-collapse-1}--\eqref{eq:amort-collapse-2} means that the RLD Fisher information gained during each of the $n$ sequential uses of the channel $\mathcal{N}^{\bm{\theta}}_{A \rightarrow B}$ is the same. It also means that the maximum attainable RLD Fisher information value for a channel family is the same when using sequential or parallel strategies.

\subsection{Chain Rule for RLD Fisher Information Value} \label{subsec:chain-rule}

Let $\{\mathcal{N}_{A\rightarrow B}^{\theta}\}_{\theta}$ be a differentiable family of quantum channels, and let $\{\rho_{RA}^{\theta}\}_{\theta}$ be a differentiable family of quantum states on systems $RA$, where the system $R$ can be of arbitrary size. Then the following chain rule holds \cite{KW20a}: 
\begin{equation}
	\widehat{I}_{F}(\theta;\{\mathcal{N}_{A\rightarrow B}^{\theta}(\rho_{RA}^{\theta})\}_{\theta})\leq\widehat{I}_{F}(\theta;\{\mathcal{N}_{A\rightarrow B}^{\theta}\}_{\theta}) +\widehat{I}_{F}(\theta;\{\rho_{RA}^{\theta}\}_{\theta}), \label{eq:single-RLD-chain-rule}
\end{equation}
analogous to one that holds for a different distinguishability measure \cite{Fang2019a}.
The chain rule in \eqref{eq:single-RLD-chain-rule} can be generalized to the case of multiparameter estimation. To do so, we first state it in the form of an operator inequality. Then we will show that \eqref{eq:single-RLD-chain-rule} arises as a simple consequence of it.

\begin{proposition} \label{prop:chain-rule-operator-ineq} 
	Let $\{\mathcal{N}_{A\rightarrow B}^{\bm{\theta}}\}_{\bm{\theta}}$ be a
	differentiable family of quantum channels, and let $\{\rho_{RA}^{\bm{\theta}}%
	\}_{\bm{\theta}}$ be a differentiable family of quantum states. Then the
	following chain-rule operator inequality holds%
	\begin{equation} \label{app-eq:operator-chain-rule}
	\widehat{I}_F(\bm{\theta};\{\mathcal{N}_{A\rightarrow B}^{\bm{\theta}}(\rho
	_{RA}^{\bm{\theta}})\}_{\bm{\theta}})\leq\sum_{j,k=1}^{D}|j\rangle\!\langle
	k|\operatorname{Tr}[(\rho_{S}^{\bm{\theta}})^{T}\operatorname{Tr}%
	_{B}[(\partial_{\theta_{j}}\Gamma_{SB}^{\mathcal{N}^{\bm{\theta}}}%
	)(\Gamma_{SB}^{\mathcal{N}^{\bm{\theta}}})^{-1}(\partial_{\theta_{k}}%
	\Gamma_{SB}^{\mathcal{N}^{\bm{\theta}}})]
	+\widehat{I}_F(\bm{\theta};\{\rho_{RA}^{\bm{\theta}}\}_{\bm{\theta}}),
	\end{equation}
	where $\rho_S^{\bm{\theta}}$ is equal to the reduced state of $\rho_{RA}^{\bm{\theta}}$ on system $A$ and system $S$ is isomorphic to system $A$.
\end{proposition}

\begin{proof}
In the proof, we make use of the following identity:
\begin{equation}
\mathcal{N}_{A\rightarrow B}^{\bm{\theta}}(\rho_{RA}^{\bm{\theta}}) = \langle\Gamma|_{AS}\rho_{RA}^{\bm{\theta}}\otimes\Gamma_{SB}^{\mathcal{N}%
		^{\bm{\theta}}}|\Gamma\rangle_{AS},
\end{equation}
where
\begin{equation}
|\Gamma\rangle_{AS} \coloneqq \sum_{i} \ket{i}_A \ket{i}_S.
\end{equation}

Consider that%
	\begin{equation}
	\widehat{I}_F(\bm{\theta};\{\mathcal{N}_{A\rightarrow B}^{\bm{\theta}}(\rho
	_{RA}^{\bm{\theta}})\}_{\bm{\theta}})=\operatorname{Tr}_{2}\!\left[
	\sum_{j,k=1}^{D}|j\rangle\!\langle k|\otimes(\partial_{\theta_{j}}%
	\mathcal{N}_{A\rightarrow B}^{\bm{\theta}}(\rho_{RA}^{\bm{\theta}}%
	))(\mathcal{N}_{A\rightarrow B}^{\bm{\theta}}(\rho_{RA}^{\bm{\theta}}%
	))^{-1}(\partial_{\theta_{k}}\mathcal{N}_{A\rightarrow B}^{\bm{\theta}}%
	(\rho_{RA}^{\bm{\theta}}))\right]  .
	\end{equation}
	
We can, with a series of manipulations, show that
\begin{align}
	&  \sum_{j,k=1}^{D}|j\rangle\!\langle k|\otimes(\partial_{\theta_{j}}%
	\mathcal{N}_{A\rightarrow B}^{\bm{\theta}}(\rho_{RA}^{\bm{\theta}}%
	))(\mathcal{N}_{A\rightarrow B}^{\bm{\theta}}(\rho_{RA}^{\bm{\theta}}%
	))^{-1}(\partial_{\theta_{k}}\mathcal{N}_{A\rightarrow B}^{\bm{\theta}}%
	(\rho_{RA}^{\bm{\theta}}))\nonumber\\
	&  =\sum_{j,k=1}^{D}|j\rangle\!\langle k|\otimes(\partial_{\theta_{j}}%
	\langle\Gamma|_{AS}\rho_{RA}^{\bm{\theta}}\otimes\Gamma_{SB}^{\mathcal{N}%
		^{\bm{\theta}}}|\Gamma\rangle_{AS})(\langle\Gamma|_{AS}\rho_{RA}%
	^{\bm{\theta}}\otimes\Gamma_{SB}^{\mathcal{N}^{\bm{\theta}}}|\Gamma
	\rangle_{AS})^{-1}(\partial_{\theta_{k}}\langle\Gamma|_{AS}\rho_{RA}%
	^{\bm{\theta}}\otimes\Gamma_{SB}^{\mathcal{N}^{\bm{\theta}}}|\Gamma
	\rangle_{AS})\\
	&  =\langle\Gamma|_{AS}\left(  \sum_{j=1}^{D}|j\rangle\!\langle0|\otimes
	\partial_{\theta_{j}}(\rho_{RA}^{\bm{\theta}}\otimes\Gamma_{SB}^{\mathcal{N}%
		^{\bm{\theta}}})\right)  |\Gamma\rangle_{AS}(\langle\Gamma|_{AS}%
	(|0\rangle\!\langle0|\otimes\rho_{RA}^{\bm{\theta}}\otimes\Gamma_{SB}%
	^{\mathcal{N}^{\bm{\theta}}})|\Gamma\rangle_{AS})^{-1}\nonumber\\
	&  \qquad\times\langle\Gamma|_{AS}\left(  \sum_{k=1}^{D}|0\rangle\!\langle
	k|\otimes\partial_{\theta_{k}}(\rho_{RA}^{\bm{\theta}}\otimes\Gamma
	_{SB}^{\mathcal{N}^{\bm{\theta}}})\right)  |\Gamma\rangle_{AS}.
	\end{align}
	
	Then identifying
	\begin{align}
	X &  =\sum_{k=1}^{D}|0\rangle\!\langle k|\otimes\partial_{\theta_{k}}(\rho
	_{RA}^{\bm{\theta}}\otimes\Gamma_{SB}^{\mathcal{N}^{\bm{\theta}}}),\\
	Y &  =|0\rangle\!\langle0|\otimes\rho_{RA}^{\bm{\theta}}\otimes\Gamma
	_{SB}^{\mathcal{N}^{\bm{\theta}}}\\
	L &  =I_{D}\otimes\langle\Gamma|_{AS}\otimes I_{RB},
	\end{align}
	we can apply the well known transformer inequality $LX^{\dag}L^{\dag}(LYL^{\dag})^{-1}LXL^{\dag}\leq LX^{\dag}Y^{-1}XL^{\dag}$ (see, e.g., \cite[Lemma~54]{KW20a} as well as \cite{Fang2019a}) to find that the last line above is
	not larger than the following one in the operator-inequality sense:%
	\begin{align}
	&  \langle\Gamma|_{AS}\left(  \sum_{j=1}^{D}|j\rangle\!\langle0|\otimes
	\partial_{\theta_{j}}(\rho_{RA}^{\bm{\theta}}\otimes\Gamma_{SB}^{\mathcal{N}%
		^{\bm{\theta}}})\right)  (|0\rangle\!\langle0|\otimes\rho_{RA}^{\bm{\theta}}%
	\otimes\Gamma_{SB}^{\mathcal{N}^{\bm{\theta}}})^{-1}\left(  \sum_{k=1}%
	^{D}|0\rangle\!\langle k|\otimes\partial_{\theta_{k}}(\rho_{RA}^{\bm{\theta}}%
	\otimes\Gamma_{SB}^{\mathcal{N}^{\bm{\theta}}})\right)  |\Gamma\rangle
	_{AS}\nonumber\\
	&  =\langle\Gamma|_{AS}\sum_{j,k=1}^{D}|j\rangle\!\langle k|\otimes
	\partial_{\theta_{j}}(\rho_{RA}^{\bm{\theta}}\otimes\Gamma_{SB}^{\mathcal{N}%
		^{\bm{\theta}}})(\rho_{RA}^{\bm{\theta}}\otimes\Gamma_{SB}^{\mathcal{N}%
		^{\bm{\theta}}})^{-1}\partial_{\theta_{k}}(\rho_{RA}^{\bm{\theta}}%
	\otimes\Gamma_{SB}^{\mathcal{N}^{\bm{\theta}}})|\Gamma\rangle_{AS}\\
	&  =\sum_{j,k=1}^{D}|j\rangle\!\langle k|\otimes\langle\Gamma|_{AS}%
	\partial_{\theta_{j}}(\rho_{RA}^{\bm{\theta}}\otimes\Gamma_{SB}^{\mathcal{N}%
		^{\bm{\theta}}})(\rho_{RA}^{\bm{\theta}}\otimes\Gamma_{SB}^{\mathcal{N}%
		^{\bm{\theta}}})^{-1}\partial_{\theta_{k}}(\rho_{RA}^{\bm{\theta}}%
	\otimes\Gamma_{SB}^{\mathcal{N}^{\bm{\theta}}})|\Gamma\rangle_{AS}.
	\end{align}
	
	Consider that%
	\begin{equation}
	\partial_{\theta_{j}}(\rho_{RA}^{\bm{\theta}}\otimes\Gamma_{SB}^{\mathcal{N}%
		^{\bm{\theta}}})=(\partial_{\theta_{j}}\rho_{RA}^{\bm{\theta}})\otimes
	\Gamma_{SB}^{\mathcal{N}^{\bm{\theta}}}+\rho_{RA}^{\bm{\theta}}\otimes
	(\partial_{\theta_{j}}\Gamma_{SB}^{\mathcal{N}^{\bm{\theta}}}).
	\end{equation}
	Then we find that%
	\begin{align}
	&  \partial_{\theta_{j}}(\rho_{RA}^{\bm{\theta}}\otimes\Gamma_{SB}%
	^{\mathcal{N}^{\bm{\theta}}})(\rho_{RA}^{\bm{\theta}}\otimes\Gamma
	_{SB}^{\mathcal{N}^{\bm{\theta}}})^{-1}\nonumber\\
	&  =((\partial_{\theta_{j}}\rho_{RA}^{\bm{\theta}})\otimes\Gamma
	_{SB}^{\mathcal{N}^{\bm{\theta}}}+\rho_{RA}^{\bm{\theta}}\otimes
	(\partial_{\theta_{j}}\Gamma_{SB}^{\mathcal{N}^{\bm{\theta}}}))(\rho
	_{RA}^{\bm{\theta}}\otimes\Gamma_{SB}^{\mathcal{N}^{\bm{\theta }}})^{-1}\\
	&  =((\partial_{\theta_{j}}\rho_{RA}^{\bm{\theta}})\otimes\Gamma
	_{SB}^{\mathcal{N}^{\bm{\theta}}}+\rho_{RA}^{\bm{\theta}}\otimes
	(\partial_{\theta_{j}}\Gamma_{SB}^{\mathcal{N}^{\bm{\theta}}}))((\rho
	_{RA}^{\bm{\theta}})^{-1}\otimes(\Gamma_{SB}^{\mathcal{N}^{\bm{\theta}}}%
	)^{-1})\\
	&  =(\partial_{\theta_{j}}\rho_{RA}^{\bm{\theta}})(\rho_{RA}^{\bm{\theta}}%
	)^{-1}\otimes\Gamma_{SB}^{\mathcal{N}^{\bm{\theta}}}(\Gamma_{SB}%
	^{\mathcal{N}^{\bm{\theta}}})^{-1}+\rho_{RA}^{\bm{\theta }}(\rho
	_{RA}^{\bm{\theta}})^{-1}\otimes(\partial_{\theta_{j}}\Gamma_{SB}%
	^{\mathcal{N}^{\bm{\theta}}})(\Gamma_{SB}^{\mathcal{N}^{\bm{\theta}}})^{-1}\\
	&  =(\partial_{\theta_{j}}\rho_{RA}^{\bm{\theta}})(\rho_{RA}^{\bm{\theta}}%
	)^{-1}\otimes\Pi_{\Gamma^{\mathcal{N}^{\bm{\theta}}}}+\Pi_{\rho^{\bm{\theta}}%
	}\otimes(\partial_{\theta_{j}}\Gamma_{SB}^{\mathcal{N}^{\bm{\theta}}}%
	)(\Gamma_{SB}^{\mathcal{N}^{\bm{\theta}}})^{-1}.
	\end{align}
	
	Right multiplying this last line by $\partial_{\theta_{k}}(\rho_{RA}%
	^{\bm{\theta}}\otimes\Gamma_{SB}^{\mathcal{N}^{\bm{\theta}}})$ gives%
	\begin{multline}
	  ((\partial_{\theta_{j}}\rho_{RA}^{\bm{\theta}})(\rho_{RA}^{\bm{\theta}}%
	)^{-1}\otimes\Pi_{\Gamma^{\mathcal{N}^{\bm{\theta}}}}+\Pi_{\rho^{\bm{\theta}}%
	}\otimes(\partial_{\theta_{j}}\Gamma_{SB}^{\mathcal{N}^{\bm{\theta}}}%
	)(\Gamma_{SB}^{\mathcal{N}^{\bm{\theta}}})^{-1})\partial_{\theta_{k}}%
	(\rho_{RA}^{\bm{\theta}}\otimes\Gamma_{SB}^{\mathcal{N}^{\bm{\theta}}%
	})\\
	  =(\partial_{\theta_{j}}\rho_{RA}^{\bm{\theta}})(\rho_{RA}^{\bm{\theta}}%
	)^{-1}(\partial_{\theta_{k}}\rho_{RA}^{\bm{\theta}})\otimes\Gamma
	_{SB}^{\mathcal{N}^{\bm{\theta}}}+(\partial_{\theta_{j}}\rho_{RA}%
	^{\bm{\theta}})\Pi_{\rho^{\bm{\theta}}}\otimes\Pi_{\Gamma^{\mathcal{N}%
			^{\bm{\theta}}}}(\partial_{\theta_{k}}\Gamma_{SB}^{\mathcal{N}^{\bm{\theta}}%
	})\\
	+\Pi_{\rho^{\bm{\theta}}}(\partial_{\theta_{k}}\rho_{RA}%
	^{\bm{\theta }})\otimes(\partial_{\theta_{j}}\Gamma_{SB}^{\mathcal{N}%
		^{\bm{\theta}}})\Pi_{\Gamma^{\mathcal{N}^{\bm{\theta}}}}+\rho_{RA}%
	^{\bm{\theta}}\otimes(\partial_{\theta_{j}}\Gamma_{SB}^{\mathcal{N}%
		^{\bm{\theta}}})(\Gamma_{SB}^{\mathcal{N}^{\bm{\theta}}})^{-1}(\partial
	_{\theta_{k}}\Gamma_{SB}^{\mathcal{N}^{\bm{\theta}}}).	
	\end{multline}
	
	Adding in the terms $(\partial_{\theta_{j}}\rho_{RA}^{\bm{\theta}})\Pi
	_{\rho^{\bm{\theta}}}^{\perp}=0$, $(\partial_{\theta_{j}}\Gamma_{SB}%
	^{\mathcal{N}^{\bm{\theta}}})\Pi_{\Gamma^{\mathcal{N}^{\bm{\theta
	}}}}^{\perp}=0$, $\Pi_{\rho^{\bm{\theta}}}^{\perp}(\partial_{\theta_{k}}%
	\rho_{RA}^{\bm{\theta}})=0$, and $\Pi_{\Gamma^{\mathcal{N}^{\bm{\theta}}}%
	}^{\perp}(\partial_{\theta_{k}}\Gamma_{SB}^{\mathcal{N}^{\bm{\theta}}})=0$,
	which follow from the support conditions for finite RLD Fisher information, the last line above becomes as follows:%
	\begin{multline}
	=(\partial_{\theta_{j}}\rho_{RA}^{\bm{\theta}})(\rho_{RA}^{\bm{\theta }}%
	)^{-1}(\partial_{\theta_{k}}\rho_{RA}^{\bm{\theta}})\otimes\Gamma
	_{SB}^{\mathcal{N}^{\bm{\theta}}}+(\partial_{\theta_{j}}\rho_{RA}%
	^{\bm{\theta}})\otimes(\partial_{\theta_{k}}\Gamma_{SB}^{\mathcal{N}%
		^{\bm{\theta}}})\\
	+(\partial_{\theta_{k}}\rho_{RA}^{\bm{\theta}})\otimes(\partial_{\theta_{j}%
	}\Gamma_{SB}^{\mathcal{N}^{\bm{\theta}}})+\rho_{RA}^{\bm{\theta}}%
	\otimes(\partial_{\theta_{j}}\Gamma_{SB}^{\mathcal{N}^{\bm{\theta}}}%
	)(\Gamma_{SB}^{\mathcal{N}^{\bm{\theta}}})^{-1}(\partial_{\theta_{k}}%
	\Gamma_{SB}^{\mathcal{N}^{\bm{\theta}}}).
	\end{multline}
	
	So then the relevant matrix is simplified as follows:%
	\begin{multline}
	\sum_{j,k=1}^{D}|j\rangle\!\langle k|\otimes\langle\Gamma|_{AS}[(\partial
	_{\theta_{j}}\rho_{RA}^{\bm{\theta}})(\rho_{RA}^{\bm{\theta}})^{-1}%
	(\partial_{\theta_{k}}\rho_{RA}^{\bm{\theta}})\otimes\Gamma_{SB}%
	^{\mathcal{N}^{\bm{\theta}}}+(\partial_{\theta_{j}}\rho_{RA}^{\bm{\theta}}%
	)\otimes(\partial_{\theta_{k}}\Gamma_{SB}^{\mathcal{N}^{\bm{\theta}}})\\
	+(\partial_{\theta_{k}}\rho_{RA}^{\bm{\theta}})\otimes(\partial_{\theta_{j}%
	}\Gamma_{SB}^{\mathcal{N}^{\bm{\theta}}})+\rho_{RA}^{\bm{\theta}}%
	\otimes(\partial_{\theta_{j}}\Gamma_{SB}^{\mathcal{N}^{\bm{\theta}}}%
	)(\Gamma_{SB}^{\mathcal{N}^{\bm{\theta}}})^{-1}(\partial_{\theta_{k}}%
	\Gamma_{SB}^{\mathcal{N}^{\bm{\theta}}})]|\Gamma\rangle_{AS}.
	\end{multline}
	Thus, we have established the following operator inequality:
	\begin{multline}
	\sum_{j,k=1}^{D}|j\rangle\!\langle k|\otimes(\partial_{\theta_{j}}%
	\mathcal{N}_{A\rightarrow B}^{\bm{\theta}}(\rho_{RA}^{\bm{\theta}}%
	))(\mathcal{N}_{A\rightarrow B}^{\bm{\theta}}(\rho_{RA}^{\bm{\theta}}%
	))^{-1}(\partial_{\theta_{k}}\mathcal{N}_{A\rightarrow B}^{\bm{\theta}}%
	(\rho_{RA}^{\bm{\theta}}))\\
	\leq\sum_{j,k=1}^{D}|j\rangle\!\langle k|\otimes\langle\Gamma|_{AS}%
	[(\partial_{\theta_{j}}\rho_{RA}^{\bm{\theta}})(\rho_{RA}^{\bm{\theta}}%
	)^{-1}(\partial_{\theta_{k}}\rho_{RA}^{\bm{\theta}})\otimes\Gamma
	_{SB}^{\mathcal{N}^{\bm{\theta}}}+(\partial_{\theta_{j}}\rho_{RA}%
	^{\bm{\theta}})\otimes(\partial_{\theta_{k}}\Gamma_{SB}^{\mathcal{N}%
		^{\bm{\theta}}})\\
	+(\partial_{\theta_{k}}\rho_{RA}^{\bm{\theta}})\otimes(\partial_{\theta_{j}%
	}\Gamma_{SB}^{\mathcal{N}^{\bm{\theta}}})+\rho_{RA}^{\bm{\theta}}%
	\otimes(\partial_{\theta_{j}}\Gamma_{SB}^{\mathcal{N}^{\bm{\theta}}}%
	)(\Gamma_{SB}^{\mathcal{N}^{\bm{\theta}}})^{-1}(\partial_{\theta_{k}}%
	\Gamma_{SB}^{\mathcal{N}^{\bm{\theta}}})]|\Gamma\rangle_{AS}.
	\end{multline}
	We can then take a partial trace over the $RB\ $systems and arrive at the
	following operator inequality:%
	\begin{multline}
	\widehat{I}_F(\bm{\theta};\{\mathcal{N}_{A\rightarrow B}^{\bm{\theta}}(\rho
	_{RA}^{\bm{\theta}})\}_{\bm{\theta}})\\
	\leq\sum_{j,k=1}^{D}|j\rangle\!\langle k|~\operatorname{Tr}_{RB}[\langle
	\Gamma|_{AS}[(\partial_{\theta_{j}}\rho_{RA}^{\bm{\theta}})(\rho
	_{RA}^{\bm{\theta}})^{-1}(\partial_{\theta_{k}}\rho_{RA}^{\bm{\theta}}%
	)\otimes\Gamma_{SB}^{\mathcal{N}^{\bm{\theta}}}+(\partial_{\theta_{j}}%
	\rho_{RA}^{\bm{\theta}})\otimes(\partial_{\theta_{k}}\Gamma_{SB}%
	^{\mathcal{N}^{\bm{\theta}}})\\
	+(\partial_{\theta_{k}}\rho_{RA}^{\bm{\theta}})\otimes(\partial_{\theta_{j}%
	}\Gamma_{SB}^{\mathcal{N}^{\bm{\theta}}})+\rho_{RA}^{\bm{\theta}}%
	\otimes(\partial_{\theta_{j}}\Gamma_{SB}^{\mathcal{N}^{\bm{\theta}}}%
	)(\Gamma_{SB}^{\mathcal{N}^{\bm{\theta}}})^{-1}(\partial_{\theta_{k}}%
	\Gamma_{SB}^{\mathcal{N}^{\bm{\theta}}})]|\Gamma\rangle_{AS}].
	\end{multline}
	Now we evaluate each term on the right:%
	\begin{align}
&  \!\!\!\!\!\langle\Gamma|_{AS}\operatorname{Tr}_{RB}[(\partial_{\theta_{j}%
}\rho_{RA}^{\bm{\theta}})(\rho_{RA}^{\bm{\theta}})^{-1}(\partial_{\theta_{k}%
}\rho_{RA}^{\bm{\theta}})\otimes\Gamma_{SB}^{\mathcal{N}^{\bm{\theta}}%
}]|\Gamma\rangle_{AS}\nonumber\\
&  =\langle\Gamma|_{AS}\operatorname{Tr}_{R}[(\partial_{\theta_{j}}\rho
_{RA}^{\bm{\theta}})(\rho_{RA}^{\bm{\theta}})^{-1}(\partial_{\theta_{k}}%
\rho_{RA}^{\bm{\theta}})]\otimes\operatorname{Tr}_{B}[\Gamma_{SB}%
^{\mathcal{N}^{\bm{\theta}}}]|\Gamma\rangle_{AS}\\
&  =\langle\Gamma|_{AS}\operatorname{Tr}_{R}[(\partial_{\theta_{j}}\rho
_{RA}^{\bm{\theta}})(\rho_{RA}^{\bm{\theta}})^{-1}(\partial_{\theta_{k}}%
\rho_{RA}^{\bm{\theta}})]\otimes I_{S}|\Gamma\rangle_{AS}\\
&  =\operatorname{Tr}[(\partial_{\theta_{j}}\rho_{RA}^{\bm{\theta}})(\rho
_{RA}^{\bm{\theta}})^{-1}(\partial_{\theta_{k}}\rho_{RA}^{\bm{\theta}})],\\
&  \!\!\!\!\!\langle\Gamma|_{AS}\operatorname{Tr}_{RB}[(\partial_{\theta_{k}%
}\rho_{RA}^{\bm{\theta}})\otimes(\partial_{\theta_{j}}\Gamma_{SB}%
^{\mathcal{N}^{\bm{\theta}}})]|\Gamma\rangle_{AS}\nonumber\\
&  =\langle\Gamma|_{AS}\operatorname{Tr}_{R}[\partial_{\theta_{k}}\rho
_{RA}^{\bm{\theta}}]\otimes\operatorname{Tr}_{B}[\partial_{\theta_{j}}%
\Gamma_{SB}^{\mathcal{N}^{\bm{\theta}}}]|\Gamma\rangle_{AS}\\
&  =\langle\Gamma|_{AS}\operatorname{Tr}_{R}[\partial_{\theta_{k}}\rho
_{RA}^{\bm{\theta}}]\otimes\partial_{\theta_{j}}\operatorname{Tr}_{B}%
[\Gamma_{SB}^{\mathcal{N}^{\bm{\theta}}}]|\Gamma\rangle_{AS}\\
&  =\langle\Gamma|_{AS}\operatorname{Tr}_{R}[\partial_{\theta_{k}}\rho
_{RA}^{\bm{\theta}}]\otimes(\partial_{\theta_{j}}I_{S})|\Gamma\rangle_{AS}\\
&  =0, \\
&  \!\!\!\!\!\langle\Gamma|_{AS}\operatorname{Tr}_{RB}[(\partial_{\theta_{j}%
}\rho_{RA}^{\bm{\theta}})\otimes(\partial_{\theta_{k}}\Gamma_{SB}%
^{\mathcal{N}^{\bm{\theta}}})]|\Gamma\rangle_{AS}=0, \text{and} \\
&  \!\!\!\!\!\langle\Gamma|_{AS}\operatorname{Tr}_{RB}[\rho_{RA}%
^{\bm{\theta}}\otimes(\partial_{\theta_{j}}\Gamma_{SB}^{\mathcal{N}%
^{\bm{\theta}}})(\Gamma_{SB}^{\mathcal{N}^{\bm{\theta}}})^{-1}(\partial
_{\theta_{k}}\Gamma_{SB}^{\mathcal{N}^{\bm{\theta}}})]|\Gamma\rangle
_{AS}\nonumber\\
&  =\langle\Gamma|_{AS}\rho_{A}^{\bm{\theta}}\otimes\operatorname{Tr}%
_{B}[(\partial_{\theta_{j}}\Gamma_{SB}^{\mathcal{N}^{\bm{\theta}}}%
)(\Gamma_{SB}^{\mathcal{N}^{\bm{\theta}}})^{-1}(\partial_{\theta_{k}}%
\Gamma_{SB}^{\mathcal{N}^{\bm{\theta}}})]|\Gamma\rangle_{AS}\\
&  =\operatorname{Tr}[(\rho_{S}^{\bm{\theta}})^{T}\operatorname{Tr}%
_{B}[(\partial_{\theta_{j}}\Gamma_{SB}^{\mathcal{N}^{\bm{\theta}}}%
)(\Gamma_{SB}^{\mathcal{N}^{\bm{\theta}}})^{-1}(\partial_{\theta_{k}}%
\Gamma_{SB}^{\mathcal{N}^{\bm{\theta}}})].
\end{align}
	
	Substituting back above, we find that
	\begin{align}
	& \widehat{I}_F(\bm{\theta};\{\mathcal{N}_{A\rightarrow B}^{\bm{\theta}}%
	(\rho_{RA}^{\bm{\theta}})\}_{\bm{\theta}})\nonumber\\
	& \leq\sum_{j,k=1}^{D}|j\rangle\!\langle k|\ \operatorname{Tr}[(\partial
	_{\theta_{j}}\rho_{RA}^{\bm{\theta}})(\rho_{RA}^{\bm{\theta}})^{-1}%
	(\partial_{\theta_{k}}\rho_{RA}^{\bm{\theta}})]\nonumber\\
	& \qquad+\sum_{j,k=1}^{D}|j\rangle\!\langle k|\ \operatorname{Tr}[(\rho
	_{S}^{\bm{\theta}})^{T}\operatorname{Tr}_{B}[(\partial_{\theta_{j}}\Gamma
	_{SB}^{\mathcal{N}^{\bm{\theta}}})(\Gamma_{SB}^{\mathcal{N}^{\bm{\theta}}%
	})^{-1}(\partial_{\theta_{k}}\Gamma_{SB}^{\mathcal{N}^{\bm{\theta}}})]\\
	& =\widehat{I}_F(\bm{\theta};\{\rho_{RA}^{\bm{\theta}}\}_{\bm{\theta}}%
	)+\sum_{j,k=1}^{D}|j\rangle\!\langle k|\ \operatorname{Tr}[(\rho_{S}%
	^{\bm{\theta}})^{T}\operatorname{Tr}_{B}[(\partial_{\theta_{j}}\Gamma
	_{SB}^{\mathcal{N}^{\bm{\theta}}})(\Gamma_{SB}^{\mathcal{N}^{\bm{\theta}}%
	})^{-1}(\partial_{\theta_{k}}\Gamma_{SB}^{\mathcal{N}^{\bm{\theta}}})].
	\end{align}
	This concludes the proof.
\end{proof}
\bigskip

Now we show how the operator-inequality chain rule leads us to the desired result in  \eqref{eq:multi-RLD-chain-rule}. 

\begin{proposition} \label{prop:rld-chain-rule}
With $W$ a weight matrix as defined previously, let $\{\mathcal{N}_{A\rightarrow B}^{\bm{\theta}}\}_{\bm{\theta}}$ be a differentiable family of quantum channels, and let $\{\rho_{RA}^{\bm{\theta}}\}_{\bm{\theta}}$ be a differentiable family of quantum states. Then the following chain-rule inequality holds%
\begin{equation} \label{eq:multi-RLD-chain-rule}
	\widehat{I}_F(\bm{\theta},W;\{\mathcal{N}_{A\rightarrow B}^{\bm{\theta}}(\rho_{RA}^{\bm{\theta}})\}_{\bm{\theta}})\leq  \widehat{I}_F(\bm{\theta},W;\{\mathcal{N}_{A\rightarrow B}^{\bm{\theta}}\}_{\bm{\theta}}) +\widehat{I}_F(\bm{\theta},W;\{\rho_{RA}^{\bm{\theta}}\}_{\bm{\theta}}).
\end{equation}
\end{proposition}

\begin{proof}
We start by restating the operator-inequality chain rule in \eqref{app-eq:operator-chain-rule}:
\begin{equation}
\widehat{I}_F(\bm{\theta};\{\mathcal{N}_{A\rightarrow B}^{\bm{\theta}}(\rho
_{RA}^{\bm{\theta}})\}_{\bm{\theta}})\leq\sum_{j,k=1}^{D}|j\rangle\!\langle
k|\operatorname{Tr}[(\rho_{S}^{\bm{\theta}})^{T}\operatorname{Tr}%
_{B}[(\partial_{\theta_{j}}\Gamma_{SB}^{\mathcal{N}^{\bm{\theta}}}%
)(\Gamma_{SB}^{\mathcal{N}^{\bm{\theta}}})^{-1}(\partial_{\theta_{k}}%
\Gamma_{SB}^{\mathcal{N}^{\bm{\theta}}})]
+\widehat{I}_F(\bm{\theta};\{\rho_{RA}^{\bm{\theta}}\}_{\bm{\theta}}).
\end{equation}
	
We sandwich the above with $W^{1/2}$ on both sides and take the trace (a positive map overall) to preserve the inequality and obtain
\begin{multline}
\Tr[ W \widehat{I}_F(\bm{\theta};\{\mathcal{N}_{A\rightarrow B}^{\bm{\theta}}(\rho
_{RA}^{\bm{\theta}})\}_{\bm{\theta}}) ] \leq
\Tr \!\left[ W\sum_{j,k=1}^{D}|j\rangle\!\langle
k|\operatorname{Tr}[(\rho_{S}^{\bm{\theta}})^{T}\operatorname{Tr}%
_{B}[(\partial_{\theta_{j}}\Gamma_{SB}^{\mathcal{N}^{\bm{\theta}}}%
)(\Gamma_{SB}^{\mathcal{N}^{\bm{\theta}}})^{-1}(\partial_{\theta_{k}}%
\Gamma_{SB}^{\mathcal{N}^{\bm{\theta}}})]]\right] \\
+ \Tr[ W\widehat{I}_F(\bm{\theta};\{\rho_{RA}^{\bm{\theta}}\}_{\bm{\theta}})].
\end{multline}	
Using the definition of the RLD Fisher information value, the above simplifies to
\begin{multline} \label{app-eq:chain-rule-intermediate}
\widehat{I}_F(\bm{\theta},W;\{\mathcal{N}_{A\rightarrow B}^{\bm{\theta}}%
(\rho_{RA}^{\bm{\theta}})\}_{\bm{\theta}}) \leq 
\Tr\! \left[ W\sum_{j,k=1}^{D}|j\rangle\!\langle
k|\operatorname{Tr}[(\rho_{S}^{\bm{\theta}})^{T}\operatorname{Tr}%
_{B}[(\partial_{\theta_{j}}\Gamma_{SB}^{\mathcal{N}^{\bm{\theta}}}%
)(\Gamma_{SB}^{\mathcal{N}^{\bm{\theta}}})^{-1}(\partial_{\theta_{k}}%
\Gamma_{SB}^{\mathcal{N}^{\bm{\theta}}})] ]\right] \\
+ \widehat{I}_F(\bm{\theta},W;\{\rho_{RA}^{\bm{\theta}}\}_{\bm{\theta}}).
\end{multline}
Now we consider that
\begin{align}
& \Tr \!\left[ W\sum_{j,k=1}^{D}|j\rangle\!\langle
k|\operatorname{Tr}[(\rho_{S}^{\bm{\theta}})^{T}\operatorname{Tr}
_{B}[(\partial_{\theta_{j}}\Gamma_{SB}^{\mathcal{N}^{\bm{\theta}}}
)(\Gamma_{SB}^{\mathcal{N}^{\bm{\theta}}})^{-1}(\partial_{\theta_{k}}
\Gamma_{SB}^{\mathcal{N}^{\bm{\theta}}})]\right] \notag \\
& \qquad = \sum_{j,k=1}^D \langle k | W | j \rangle \Tr\!\left[ (\rho_{S}^{\bm{\theta}})^{T} \operatorname{Tr}%
_{B}[(\partial_{\theta_{j}}\Gamma_{SB}^{\mathcal{N}^{\bm{\theta}}}%
)(\Gamma_{SB}^{\mathcal{N}^{\bm{\theta}}})^{-1}(\partial_{\theta_{k}}%
\Gamma_{SB}^{\mathcal{N}^{\bm{\theta}}})]\right] \\
& \qquad = \Tr\!\left[ (\rho_{S}^{\bm{\theta}})^{T} \sum_{j,k=1}^D \langle k | W | j \rangle \operatorname{Tr}%
_{B}[(\partial_{\theta_{j}}\Gamma_{SB}^{\mathcal{N}^{\bm{\theta}}}%
)(\Gamma_{SB}^{\mathcal{N}^{\bm{\theta}}})^{-1}(\partial_{\theta_{k}}%
\Gamma_{SB}^{\mathcal{N}^{\bm{\theta}}})]  \right] \\
& \qquad \leq \sup_{\rho^{\bm{\theta}}_S }  \Tr \!\left[ (\rho_{S}^{\bm{\theta}})^{T} \sum_{j,k=1}^D \langle k | W | j \rangle \operatorname{Tr}%
_{B}[(\partial_{\theta_{j}}\Gamma_{SB}^{\mathcal{N}^{\bm{\theta}}}%
)(\Gamma_{SB}^{\mathcal{N}^{\bm{\theta}}})^{-1}(\partial_{\theta_{k}}%
\Gamma_{SB}^{\mathcal{N}^{\bm{\theta}}})]  \right] \\
& \qquad =  \left\Vert \sum_{j,k=1}^D \langle k | W | j \rangle \operatorname{Tr}%
_{B}[(\partial_{\theta_{j}}\Gamma_{SB}^{\mathcal{N}^{\bm{\theta}}}%
)(\Gamma_{SB}^{\mathcal{N}^{\bm{\theta}}})^{-1}(\partial_{\theta_{k}}%
\Gamma_{SB}^{\mathcal{N}^{\bm{\theta}}})]   \right\Vert_\infty \\
& \qquad = \widehat{I}_F
(\bm{\theta},W;\{\mathcal{N}_{A\rightarrow B}^{\bm{\theta}}\}_{\bm{\theta}}. \
\end{align}

Using the above and \eqref{app-eq:chain-rule-intermediate}, we conclude that
\begin{equation}
\widehat{I}_F(\bm{\theta},W;\{\mathcal{N}_{A\rightarrow B}^{\bm{\theta}}%
(\rho_{RA}^{\bm{\theta}})\}_{\bm{\theta}}) \leq \widehat{I}_F
(\bm{\theta},W;\{\mathcal{N}_{A\rightarrow B}^{\bm{\theta}}\}_{\bm{\theta}}
+ \widehat{I}_F(\bm{\theta},W;\{\rho_{RA}^{\bm{\theta}}\}_{\bm{\theta}}).
\end{equation}
This concludes the proof.
\end{proof}

We now show how the chain rule results in an amortization collapse. The chain rule in \eqref{eq:multi-RLD-chain-rule} can be rewritten as
\begin{equation}
	\widehat{I}_F(\bm{\theta},W;\{\mathcal{N}_{A\rightarrow B}^{\bm{\theta}}(\rho_{RA}^{\bm{\theta}})\}_{\bm{\theta}}) - \widehat{I}_F(\bm{\theta},W;\{\rho_{RA}^{\bm{\theta}}\}_{\bm{\theta}}) \leq  \widehat{I}_F(\bm{\theta},W;\{\mathcal{N}_{A\rightarrow B}^{\bm{\theta}}\}_{\bm{\theta}}).
\end{equation}
Taking the supremum over input-state families on the left-hand side gives
\begin{equation}
	\widehat{I}_F^{\mathcal{A}}(\bm{\theta},W;\{\mathcal{N}_{A\rightarrow B}^{\bm{\theta}}\}_{\bm{\theta}}) \leq \widehat{I}_F(\bm{\theta},W;\{\mathcal{N}_{A\rightarrow B}^{\bm{\theta}}\}_{\bm{\theta}}).
\end{equation}
When the above is combined with \eqref{eq:amortization-basic-inequality}, we obtain the desired amortization collapse in \eqref{eq:amort-collapse-2}.

We thus see that the chain rule for the multiparameter RLD Fisher information \eqref{eq:multi-RLD-chain-rule} leads to the amortization collapse in \eqref{eq:amort-collapse-2}.
The amortization collapse implies that the $n$-sequential-use RLD Fisher information is simply $n$ times the single-use RLD Fisher information.

\section{Single-letter Quantum Cramer--Rao Bound}

The following single-letter Cramer--Rao bound for estimating parameter $\theta$ embedded in $\{\mathcal{N}_{A\rightarrow B}^{\theta}\}_{\theta}$ holds both in the parallel setting \cite{Hayashi2011} and the more general sequential setting \cite{KW20a} of channel estimation. In the sequential setting, it arises from the combination of the chain rule, the ensuing amortization collapse for the RLD Fisher information, and Theorem~18 of \cite{KW20a}:
\begin{equation}
	\operatorname{Var}(\hat{\theta})\geq\frac{1}{n\widehat{I}_{F}(\theta;\{\mathcal{N}_{A\rightarrow B}^{\theta}\}_{\theta})}. \label{eq:single-letter-single-param-crb}
\end{equation}

To prove an analogous result for the multiparameter case, we need the following converse theorem, which is a generalization of \cite[Theorem~18]{KW20a}:

\begin{theorem} \label{thm:meta-converse}
Consider a general sequential estimation protocol of the form depicted in Figure~\ref{fig:sequential strategies}. The following inequality holds:
\begin{equation}
\widehat{I}_{F}(\bm{\theta}, W;\{\omega_{R_{n}B_{n}}^{\bm{\theta}}\}_{\bm{\theta}})\leq
n\cdot \widehat{I}_{F}^{\mathcal{A}}(\bm{\theta}, W;\{\mathcal{N}_{A\rightarrow B}^{\bm{\theta}
}\}_{\bm{\theta}}),
\end{equation}
where $\omega_{R_{n}B_{n}}^{\bm{\theta}}$ is the final state of an $n$-round sequential estimation protocol, and $\mathcal{S}^1$ through $\mathcal{S}^{n-1}$ are interleaving quantum channels in the protocol, both as in Figure \ref{fig:sequential strategies}.
\end{theorem}

\begin{proof}
Consider that
\begin{align}
&  \widehat{I}_{F}(\bm{\theta}, W;\{\omega_{R_{n}B_{n}}^{\bm{\theta}}\}_{\bm{\theta}})\nonumber\\
&  =\widehat{I}_{F}(\bm{\theta}, W;\{\omega_{R_{n}B_{n}}^{\bm{\theta}}\}_{\bm{\theta}
})-\widehat{I}_{F}(\bm{\theta}, W;\{\rho_{R_{1}A_{1}}\}_{\bm{\theta}})\\
&  =\widehat{I}_{F}(\bm{\theta}, W;\{\omega_{R_{n}B_{n}}^{\bm{\theta}}\}_{\bm{\theta}
})-\widehat{I}_{F}(\bm{\theta}, W;\{\rho_{R_{1}A_{1}}\}_{\bm{\theta}})+\sum_{i=2}^{n}\left(
\widehat{I}_{F}(\bm{\theta}, W;\{\rho_{R_{i}A_{i}}^{\bm{\theta}}\}_{\bm{\theta}})-\widehat{I}
_{F}(\bm{\theta}, W;\{\rho_{R_{i}A_{i}}^{\bm{\theta}}\}_{\bm{\theta}})\right) \\
&  =\widehat{I}_{F}(\bm{\theta}, W;\{\omega_{R_{n}B_{n}}^{\bm{\theta}}\}_{\bm{\theta}
})-\widehat{I}_{F}(\bm{\theta}, W;\{\rho_{R_{1}A_{1}}\}_{\bm{\theta}})\nonumber\\
&  \qquad+\sum_{i=2}^{n}\left(  \widehat{I}_{F}(\bm{\theta}, W;\{\mathcal{S}
_{R_{i-1}B_{i-1}\rightarrow R_{i}A_{i}}^{i-1}(\rho_{R_{i-1}B_{i-1}}^{\bm{\theta}
})\}_{\bm{\theta}})-\widehat{I}_{F}(\bm{\theta}, W;\{\rho_{R_{i}A_{i}}^{\bm{\theta}}\}_{\bm{\theta}
})\right) \\
&  \leq\widehat{I}_{F}(\bm{\theta}, W;\{\omega_{R_{n}B_{n}}^{\bm{\theta}}\}_{\bm{\theta}
})-\widehat{I}_{F}(\bm{\theta}, W;\{\rho_{R_{1}A_{1}}\}_{\bm{\theta}})\nonumber\\
&  \qquad+\sum_{i=2}^{n}\left(  \widehat{I}_{F}(\bm{\theta}, W;\{\rho_{R_{i-1}B_{i-1}
}^{\bm{\theta}}\}_{\bm{\theta}})-\widehat{I}_{F}(\bm{\theta}, W;\{\rho_{R_{i}A_{i}}^{\bm{\theta}
}\}_{\bm{\theta}})\right) \\
&  =\sum_{i=1}^{n}\left(  \widehat{I}_{F}(\bm{\theta}, W;\{\rho_{R_{i}B_{i}}^{\bm{\theta}
}\}_{\bm{\theta}})-\widehat{I}_{F}(\bm{\theta}, W;\{\rho_{R_{i}A_{i}}^{\bm{\theta}}\}_{\bm{\theta}
}\right) \\
&  =\sum_{i=1}^{n}\left(  \widehat{I}_{F}(\bm{\theta}, W;\{\mathcal{N}_{A_{i}
	\rightarrow B_{i}}^{\bm{\theta}}(\rho_{R_{i}A_{i}}^{\bm{\theta}})\}_{\bm{\theta}}
)-\widehat{I}_{F}(\bm{\theta}, W;\{\rho_{R_{i}A_{i}}^{\bm{\theta}}\}_{\bm{\theta}}\right) \\
&  \leq n\cdot\sup_{\{\rho_{RA}^{\bm{\theta}}\}_{\bm{\theta}}}\left[  \widehat{I}
_{F}(\bm{\theta}, W;\{\mathcal{N}_{A\rightarrow B}^{\bm{\theta}}(\rho_{RA}^{\bm{\theta}
})\}_{\bm{\theta}})-\widehat{I}_{F}(\bm{\theta}, W;\{\rho_{RA}^{\bm{\theta}})\}_{\bm{\theta}})\right]
\\
&  =n\cdot \widehat{I}_{F}^{\mathcal{A}}(\bm{\theta}, W;\{\mathcal{N}_{A\rightarrow
	B}^{\bm{\theta}}\}_{\theta}).
\end{align}
The first equality follows because the initial state $\rho_{R_1 A_1}$ has no dependence on any of the parameters in $\bm{\theta}$. The first inequality arises due to the data-processing inequality for the RLD Fisher information. The other steps are straightforward manipulations.
\end{proof}

We have now assembled all the necessary components needed to establish our main result, a multiparameter Cramer--Rao bound for sequential channel estimation.

\begin{theorem} \label{thm:single-letter-multi-param-crb}
	For a differentiable channel family $\{\mathcal{N}_{A\rightarrow B}^{\bm{\theta}}\}_{\bm{\theta}}$ and a positive semidefinite weight matrix $W$ with $\Tr[W] = 1$, the following multiparameter Cramer--Rao bound holds:
	\begin{equation}
		\Tr[ W \text{Cov}(\bm{\theta}) ] \geq \frac{1}{n  \widehat{I}_F (\bm{\theta},W;\{ \mathcal{N}_{A\rightarrow B}^{\bm{\theta}} \}_{\bm{\theta}}) }. \label{eq:single-letter-multi-param-crb}
	\end{equation}
\end{theorem}

\begin{proof}
	The multiparameter Cramer--Rao bound follows as a direct consequence of the chain rule \eqref{eq:multi-RLD-chain-rule}, ensuing amortization collapse \eqref{eq:amort-collapse-2}, as well as the converse theorem \eqref{thm:meta-converse}.
\end{proof}


The bounds in \eqref{eq:single-letter-single-param-crb} and \eqref{eq:single-letter-multi-param-crb} have an important implication: if the RLD Fisher information is finite, then  Heisenberg scaling is impossible. That is, if $\operatorname{supp}(\partial_{\theta}\Gamma_{RB}^{\mathcal{N}^{\theta}})\subseteq\operatorname{supp}(\Gamma_{RB}^{\mathcal{N}^{\theta}})$ holds for the single-parameter case or if \eqref{eq:app:channel-RLD-finiteness-cond} holds for the multiparameter case, then all estimation strategies are limited by the shot-noise limit. A simple corollary of this is that for every full-rank quantum channel family, Heisenberg scaling is unattainable. Our no-go condition applies generally, i.e., for all quantum channels and in the sequential estimation setting. It builds upon and generalizes results on this problem from other works, which apply to the parallel estimation setting \cite{Fujiwara2008, Hayashi2011}, to sequential channel estimation with Markovian noise \cite{ Demkowicz-Dobrzanski2017, Zhou2018, Zhou2020}, and to the general sequential channel estimation setting \cite{Mat10}.

\section{Evaluating Bounds for the Generalized Amplitude Damping Channel}

Here, we apply the Cramer--Rao bound \eqref{eq:single-letter-multi-param-crb} to the task of estimating the parameters of a generalized amplitude damping channel \cite{NC00}, a qubit-to-qubit channel that is used to model noise in, e.g., superconducting quantum circuits. A generalized amplitude damping channel (GADC) is defined in terms of a loss parameter $\gamma \in (0,1)$ and a noise parameter $N \in (0,1)$. The GADC has been studied previously in the context of quantum estimation theory. In Ref.~\cite{Fujiwara_2003}, the SLD Fisher information matrix for a generalized Pauli channel was evaluated, and Ref.~\cite{Fujiwara2004} concerns optimal strategies for the estimation of the loss parameter $\gamma$ of a GADC.

The Choi operator of a GADC $\mathcal{A}_{\gamma,N}$ with loss parameter $\gamma$ and noise parameter $N$ is given by
\begin{equation}
	\Gamma_{RB}^{\mathcal{A}_{\gamma,N}}  \coloneqq 
	\begin{bmatrix}
		1-\gamma N & 0 & 0 & \sqrt{1-\gamma}\\
		0 & \gamma N & 0 & 0\\
		0 & 0 & \gamma\left(  1-N\right)  & 0\\
		\sqrt{1-\gamma} & 0 & 0 & 1-\gamma\left(  1-N\right)
	\end{bmatrix}.
\end{equation}

Suppose that we wish to estimate both parameters $\gamma$ and $N$ simultaneously. Let us choose
\begin{equation}
W = \frac{1}{4}  \begin{pmatrix}1&1\\1&3\end{pmatrix}
,
\label{eq:W-matrix-choice-end}
\end{equation}
which satisfies the requirements for a valid weight matrix (positive semi-definite with unit trace). To compute the RLD bound in \eqref{eq:single-letter-multi-param-crb}, we calculate
\begin{multline} \label{eq:gadc-rld-info-value}
	\widehat{I}_F(\{\gamma, N\}, W; \{\mathcal{A}_{\gamma, N} \}_{\gamma, N})  = \frac{1}{4} \Big\Vert \Tr_B[ (\partial_{\gamma} \Gamma^{\mathcal{A}}) (\Gamma^{\mathcal{A}})^{-1} (\partial_\gamma \Gamma^{\mathcal{A}}) + (\partial_{\gamma} \Gamma^{\mathcal{A}}) (\Gamma^{\mathcal{A}})^{-1} (\partial_N \Gamma^{\mathcal{A}}) \\ + (\partial_{N} \Gamma^{\mathcal{A}}) (\Gamma^{\mathcal{A}})^{-1} (\partial_\gamma \Gamma^{\mathcal{A}}) + 3 (\partial_{N} \Gamma^{\mathcal{A}}) (\Gamma^{\mathcal{A}})^{-1} (\partial_N \Gamma^{\mathcal{A}})]  \Big\Vert_{\infty}.
\end{multline}
where $\mathcal{A}$ is used as shorthand for $\mathcal{A}_{\gamma, N}$ and the system labels for $RB$ are omitted. As a consequence of \eqref{eq:single-letter-multi-param-crb}, the inverse of \eqref{eq:gadc-rld-info-value} is a lower bound on $\Tr[ W \text{Cov}(\{ \gamma, N \}) ] $.

We have
\begin{align}
	\partial_{\gamma}\Gamma_{RB}^{\mathcal{A}_{\gamma,N}} &=
\begin{bmatrix}
-N & 0 & 0 & -\frac{1}{2\sqrt{1-\gamma}}\\
0 & N & 0 & 0\\
0 & 0 & 1-N & 0\\
-\frac{1}{2\sqrt{1-\gamma}} & 0 & 0 & -\left(  1-N\right)
\end{bmatrix}, \text{ and} \\
	\partial_{N}\Gamma_{RB}^{\mathcal{A}_{\gamma,N}} &=-\gamma\left(  I_{2}
\otimes\sigma_{Z}\right)  .
\end{align}

We then have
\begin{align}
	\Tr_B[(\partial_{\gamma} \Gamma^{\mathcal{A}_{\gamma, N}}) (\Gamma^{\mathcal{A}_{\gamma, N}})^{-1} (\partial_\gamma \Gamma^{\mathcal{A}_{\gamma, N}})  ] &=
	\begin{bmatrix}
		\frac{\frac{1}{N-\gamma N}+\frac{1}{1-N}-4}{4 \gamma^2} & 0 \\
		0 & \frac{\frac{1}{(\gamma-1) (N-1)}+\frac{1}{N}-4}{4 \gamma^2}
	\end{bmatrix},
	\\
	\Tr_B[(\partial_{\gamma} \Gamma^{\mathcal{A}_{\gamma, N}}) (\Gamma^{\mathcal{A}_{\gamma, N}})^{-1} (\partial_N \Gamma^{\mathcal{A}_{\gamma, N}})  ] &= 
	\begin{bmatrix}
	-\frac{1-2 N}{2 \gamma N(1-N)}  & 0 \\
	0 & -\frac{1-2 N}{2 \gamma N(1-N)}
	\end{bmatrix},
	\\
	\Tr_B[(\partial_{N} \Gamma^{\mathcal{A}_{\gamma, N}}) (\Gamma^{\mathcal{A}_{\gamma, N}})^{-1} (\partial_\gamma \Gamma^{\mathcal{A}_{\gamma, N}})  ] &=
	\begin{bmatrix}
	-\frac{1-2 N}{2 \gamma N(1-N)}  & 0 \\
	0 & -\frac{1-2 N}{2 \gamma N(1-N)}
	\end{bmatrix},
	\\
	\Tr_B[(\partial_{N} \Gamma^{\mathcal{A}_{\gamma, N}}) (\Gamma^{\mathcal{A}_{\gamma, N}})^{-1} (\partial_N \Gamma^{\mathcal{A}_{\gamma, N}})  ] &= 
	\begin{bmatrix}
		\frac{1}{N(1-N)} & 0 \\
		0 & \frac{1}{N(1-N)}
	\end{bmatrix},
\end{align}
with which the expression in \eqref{eq:gadc-rld-info-value} is directly calculated. Furthermore, we verify our direct calculation by explicitly calculating the SDP in Proposition \ref{prop:rld-channels-sdp} and find agreement between the two approaches with up to eights digits of precision.

We compare the RLD Fisher information bound to the generalized Helstrom Cramer--Rao bound \cite{Albarelli2020}, which in this case reduces to the SLD Fisher information bound because it is being applied for the case of parametric, rather than semiparametric, estimation. This lower bound can be achieved asymptotically up to a constant prefactor \cite{Guta2006, Hayashi2008, Yamagata2013, Yang2019b, Tsang2019a}.

To calculate the SLD Fisher information, we choose input probe state $|\psi(p)\rangle \coloneqq \sqrt{p} \ket{00} + \sqrt{1-p} \ket{11}$, with $W$ chosen as in the RLD calculation. It suffices to optimize the single-copy SLD Fisher information over such states due to the $\sigma_Z$ covariance of the channel $\mathcal{A}_{\gamma, N}$. This SLD bound for estimating just the loss parameter $\gamma$ of a GADC was calculated in Ref.~\cite{Fujiwara2004}.

With our particular choice of input state, we have
\begin{equation}
	\rho_{\gamma, N} \coloneqq \mathcal{A}_{\gamma,N}(\psi_p) = 
	\begin{bmatrix}
		p (1 - \gamma N) & 0 & 0 & \sqrt{p(1-p)(1-\gamma)} \\
		0 & p \gamma N & 0 & 0 \\
		0 & 0 & (1-p) \gamma (1-N) & 0 \\
		\sqrt{p(1-p)(1-\gamma)} & 0 & 0 & (1-p)( 1 - (1-N)\gamma )
	\end{bmatrix}.
\end{equation}

The SLD Fisher information takes the form of a $2 \times 2$ matrix:
\begin{equation} \label{eq:app-sld-fisher-matrix}
	I_F({\gamma, N}; \{ \rho_{\gamma, N} \}_{\gamma, N})_{jk} = \Tr\!\left[ \rho_{\gamma, N} L_j L_k \right] ,
\end{equation}
where $j$ and $k$ each take values $1$ or $2$ that correspond to either $\gamma$ or $N$. We have $L_j$ defined implicitly via
$
	\partial_j \rho_{\gamma, N} = \frac{1}{2} \left( \rho_{\gamma, N} L_j + L_j \rho_{\gamma, N} \right)
$.

Consider that $\rho_{\gamma, N}$ has spectral decomposition
\begin{equation}
\rho_{\gamma, N} = \sum_j \lambda_{\gamma, N}^j | \psi_{\gamma, N}^j \rangle \langle \psi_{\gamma, N}^j |.
\end{equation}
We then have
\begin{align}
L_\gamma &= 2\sum_{j,k:\lambda^{j}_{\gamma, N}+\lambda^{k}_{\gamma, N}>0} \frac{\langle\psi^{j}_{\gamma, N}|(\partial_{\gamma}\rho_{\gamma, N})|\psi^{k}_{\gamma, N}\rangle}{\lambda^{j}_{\gamma, N}+\lambda^{k}_{\gamma, N}}|\psi^{j}_{\gamma, N}\rangle\!\langle\psi^{k}_{\gamma, N}| \text{~and} \\
L_N &= 2\sum_{j,k:\lambda^{j}_{\gamma, N}+\lambda^{k}_{\gamma, N}>0} \frac{\langle\psi^{j}_{\gamma, N}|(\partial_{N}\rho_{\gamma, N})|\psi^{k}_{\gamma, N}\rangle}{\lambda^{j}_{\gamma, N}+\lambda^{k}_{\gamma, N}}|\psi^{j}_{\gamma, N}\rangle\!\langle\psi^{k}_{\gamma, N}|,
\end{align}
which enables us to calculate the elements of the SLD Fisher information matrix given in~\eqref{eq:app-sld-fisher-matrix}. 

\begin{figure}
 	\begin{subfigure}{.5\textwidth}
 		\centering
 		\includegraphics[width=0.9\linewidth]{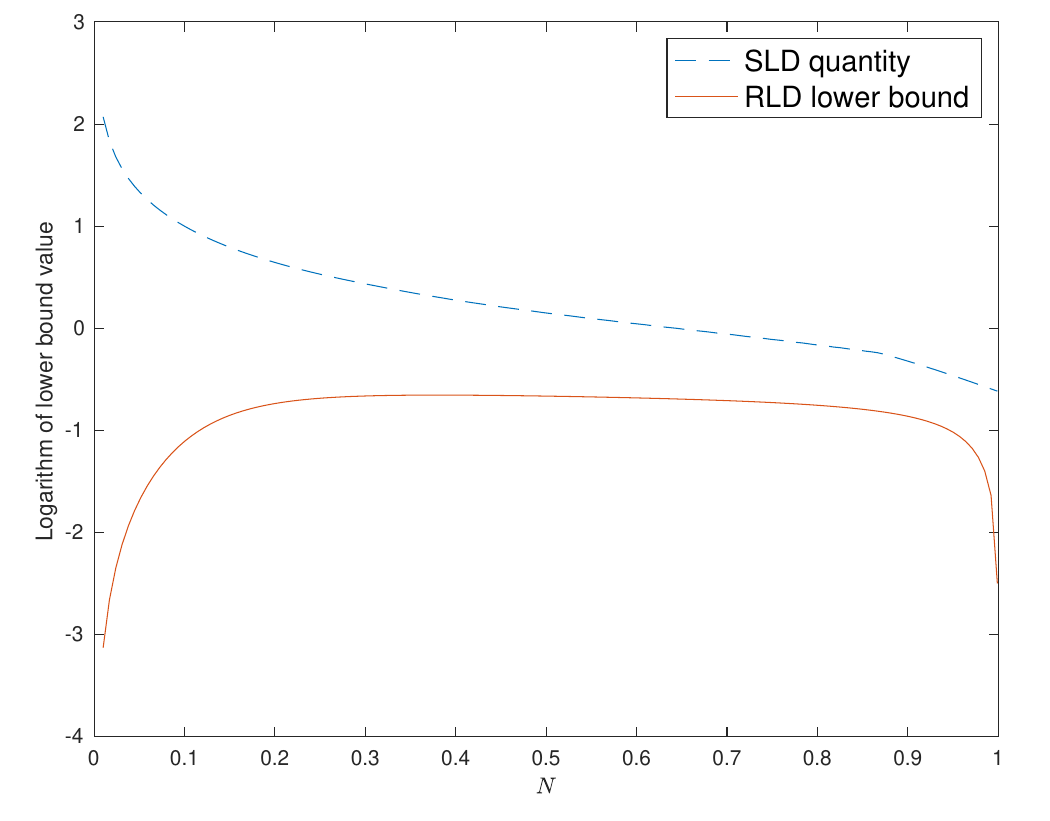}
 		\caption{}
 		\label{fig:fixed-N-0-2}
 	\end{subfigure}
	\begin{subfigure}{.5\textwidth}
		\centering
		\includegraphics[width=0.9\linewidth]{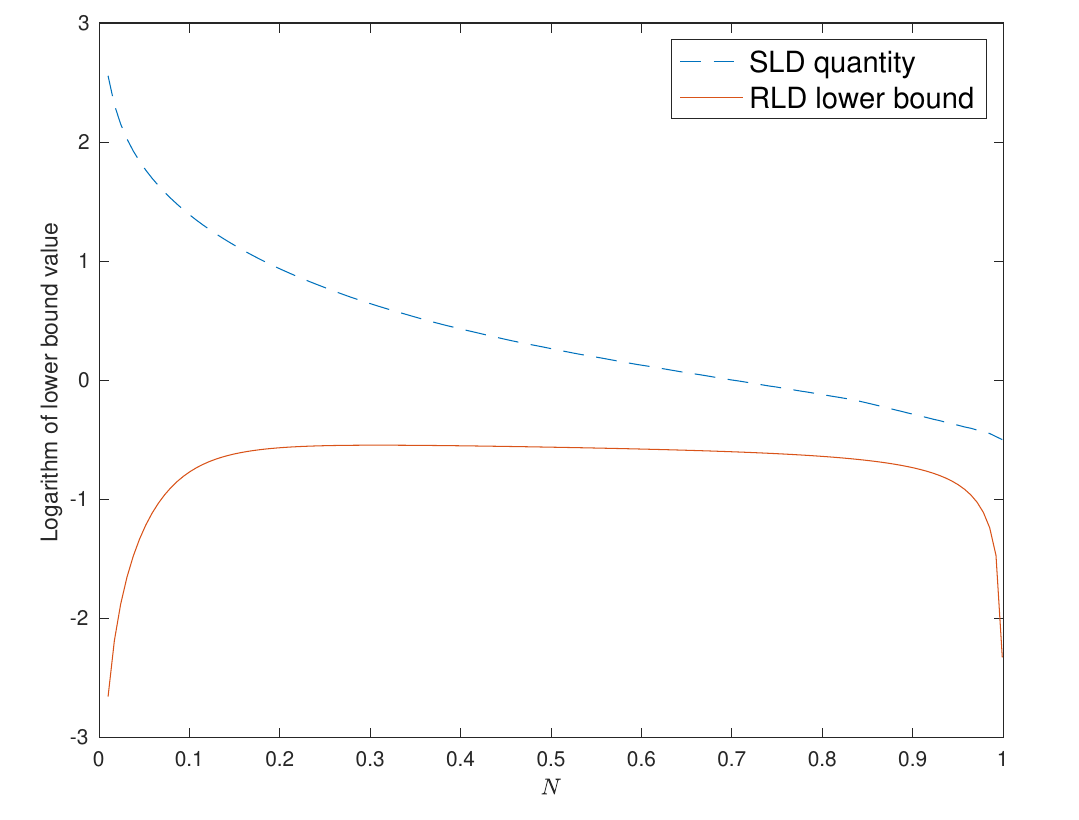}
		\caption{}
		\label{fig:fixed-N-0-3}
	\end{subfigure}
	\caption{Logarithms of the SLD quantity \eqref{eq:sld-quantity-gadc} and the inverse of the RLD Fisher information value \eqref{eq:gadc-rld-info-value} versus loss $\gamma$ with fixed noise $N$. In (a), $N=0.2$, and in (b), $N=0.3$.}
	\label{fig:gadc-estimation-two-parameters-fixed-N}
	\vspace{2em}
	\begin{subfigure}{.5\textwidth}
		\centering
		\includegraphics[width=0.9\linewidth]{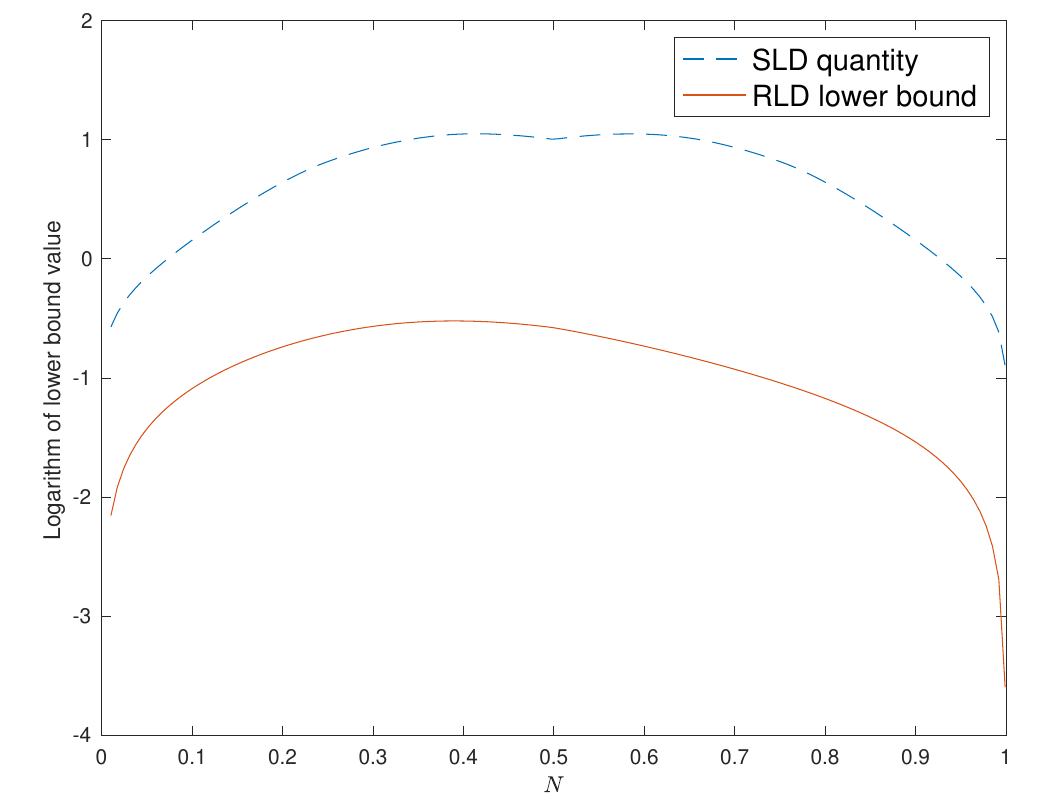}
		\caption{}
		\label{fig:fixed-g-0-2}
	\end{subfigure}
	\begin{subfigure}{.5\textwidth}
		\centering
		\includegraphics[width=0.9\linewidth]{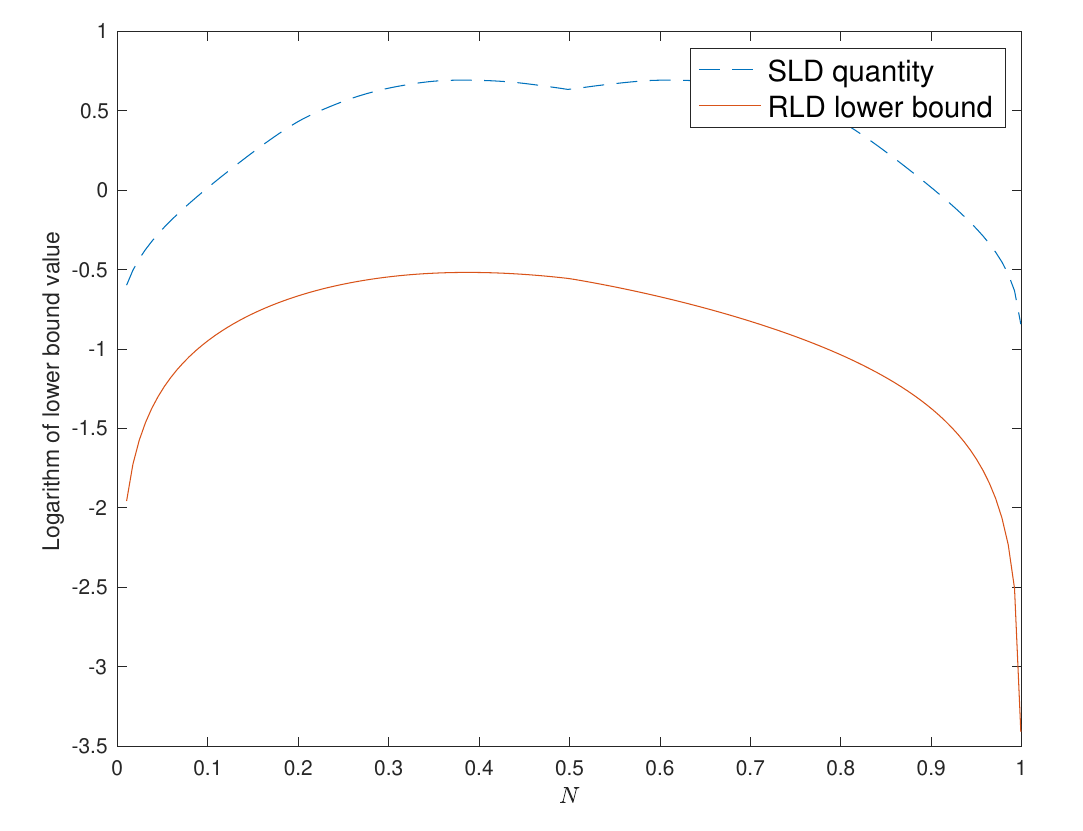}
		\caption{}
		\label{fig:fixed-g-0-3}
	\end{subfigure}
	\caption{Logarithms of the SLD quantity \eqref{eq:sld-quantity-gadc} and the inverse of the RLD Fisher information value \eqref{eq:gadc-rld-info-value} versus noise $N$ with fixed loss $\gamma$. In (a), $\gamma=0.2$, and in (b), $\gamma=0.3$.
		In each of the four figures above, both lines indicate lower bounds on the quantity
		$\Tr[ W \text{Cov} (\{ \gamma, N \} )]$
		where $W$ is chosen in \eqref{eq:W-matrix-choice-end}
		(as discussed, the SLD Fisher information value is a lower bound up to a constant prefactor).
		For the SLD quantity, we optimize over input states of the form $\sqrt{p} \ket{00}+ \sqrt{1-p} \ket{11}$.
	}
	\label{fig:gadc-estimation-two-parameters-fixed-g}
\end{figure}

We optimize over the parameter $p$ to obtain the SLD Fisher information bound. To be clear, we minimize
\begin{equation} \label{eq:sld-quantity-gadc}
\Tr[W [I_F(\{\gamma,N\}; \{\mathcal{A}_{\gamma, N}(\psi(p)))\}_{\{\gamma,N\}}]^{-1}]
\end{equation}
with respect to $p$. We compare the SLD quantity \eqref{eq:sld-quantity-gadc} to our RLD lower bound \eqref{eq:gadc-rld-info-value} in Figures~\ref{fig:gadc-estimation-two-parameters-fixed-N} and \ref{fig:gadc-estimation-two-parameters-fixed-g}. In Figure~\ref{fig:gadc-estimation-two-parameters-fixed-N}, we keep $N$ fixed and vary $\gamma$ from 0 to 1. In Figure~\ref{fig:gadc-estimation-two-parameters-fixed-g}, instead we keep $\gamma$ fixed and vary $N$ from 0 to 1. We demonstrate that the RLD lower bound is within one to two orders of magnitude of the SLD Fisher information.

\section{Summary and Conclusion}

In this paper, we provided a single-letter, efficiently computable Cramer--Rao bound for the task of multiparameter estimation. We did so by introducing the amortized RLD Fisher information and proving an amortization collapse for it. The single-letter Cramer--Rao bound also leads to a simple yet rigorous no-go condition for Heisenberg scaling in quantum multiparameter estimation. Finally, we evaluated our bound when estimating the two parameters of a generalized amplitude damping channel.

Lately, multiparameter estimation has been a fruitful area of study in quantum information theory. We believe that our results both offer a new perspective and extend the current body of work, and so they should have metrological applications theoretically and experimentally.

\section*{Acknowledgments}
This paper is dedicated to the memory of Jonathan P. Dowling. Jon's impact on quantum metrology was far-reaching, and he played a role in inspiring this project.  We acknowledge Pieter Kok, Jasminder Sidhu, and Sisi Zhou for helpful discussions. We thank Masahito Hayashi for feedback on our paper.
VK acknowledges support from the LSU Economic Development Assistantship. VK and MMW acknowledge support from the US National Science Foundation via grant number 1907615. MMW acknowledges support from Stanford QFARM and AFOSR (FA9550-19-1-0369). 

\bibliographystyle{unsrt}
\bibliography{estimation-refs}

\appendix

\section{Technical Lemma} \label{app:lemma-dual-form}

\begin{lemma}
	\label{lemma:dual-form}Let $K$ and $Z$ be Hermitian operators,
	and let $W$ be a linear operator. Then the dual of the following semi-definite
	program
	\begin{equation}
		\inf_{M}\left\{  \operatorname{Tr}[KM]:
		\begin{bmatrix}
			M & W^{\dag}\\
			W & Z
		\end{bmatrix}
		\geq0\right\}  ,
	\end{equation}
	with $M$ Hermitian, is given by
	\begin{equation}
		\sup_{P,Q,R}\left\{  2\operatorname{Re}(\operatorname{Tr}[W^{\dag
		}Q])-\operatorname{Tr}[ZR]:P\leq K,
		\begin{bmatrix}
			P & Q^{\dag}\\
			Q & R
		\end{bmatrix}
		\geq0\right\}  ,
	\end{equation}
	where $Q$ is a linear operator and $P$ and $R$ are Hermitian.
\end{lemma}

\begin{proof}
	The standard forms of a\ primal and dual semi-definite program, for $A$ and
	$B$ Hermitian and $\Phi$ a Hermiticity-preserving map, are respectively as
	follows \cite{Wat18}:
	\begin{align}
		&  \inf_{Y\geq0}\left\{  \operatorname{Tr}[BY]:\Phi^{\dag}(Y)\geq A\right\}
		,\\
		&  \sup_{X\geq0}\left\{  \operatorname{Tr}[AX]:\Phi(X)\leq B\right\}  ,
	\end{align}
	where $\Phi^{\dag}$ is the Hilbert--Schmidt adjoint of $\Phi$. Noting that
	\begin{equation}
		\begin{bmatrix}
			M & W^{\dag}\\
			W & Z
		\end{bmatrix}
		\geq0\quad\Longleftrightarrow\quad
		\begin{bmatrix}
			M & -W^{\dag}\\
			-W & Z
		\end{bmatrix}
		\geq0\quad\Longleftrightarrow\quad
		\begin{bmatrix}
			M & 0\\
			0 & 0
		\end{bmatrix}
		\geq
		\begin{bmatrix}
			0 & W^{\dag}\\
			W & -Z
		\end{bmatrix}
		,
	\end{equation}
	we conclude the statement of the lemma after making the following
	identifications:
	\begin{align}
		B  &  =K,\quad Y=M,\quad\Phi^{\dag}(M)=
		\begin{bmatrix}
			M & 0\\
			0 & 0
		\end{bmatrix}
		,\\
		A  &  =
		\begin{bmatrix}
			0 & W^{\dag}\\
			W & -Z
		\end{bmatrix}
		,\quad X=
		\begin{bmatrix}
			P & Q^{\dag}\\
			Q & R
		\end{bmatrix}
		,\quad\Phi(X)=P.
	\end{align}
	This concludes the proof.
\end{proof}

\end{document}